\documentclass{article}
\usepackage{amsmath,bm}
\usepackage{amssymb}
\usepackage[utf8]{inputenc}
\usepackage{amsfonts}
\usepackage{graphicx}
\usepackage{booktabs}
\usepackage{multirow}
\usepackage{authblk}
\usepackage{pdfpages}
\usepackage{hyperref,natbib,subfiles}
\usepackage[a4paper, total={6in, 8in}]{geometry}
\graphicspath{ {images/} }

\newcommand{\E}{\mathbb {E}}
\newcommand{\R}{\mathbb {R}}
\newcommand{\Prob}{\mathbb {P}}
\newtheorem{theorem}{Theorem}
\newtheorem{lema}{Lemma}[section]

\newcommand{\Var}{\operatorname{Var}}
\newcommand{\Cov}{\operatorname{Cov}}

\newcommand{\cov}{\operatorname{cov}}

\title{Clustering inference in multiple groups}
\author[1]{Debora Zava Bello}
\author[1,*]{ Marcio Valk }
\author[1]{Gabriela Bettella Cybis}
\affil[1]{Department of Statistics, Federal University of Rio Grande do Sul}
\affil[*]{Corresponding author: Marcio Valk, marcio.valk@ufrgs.br}

\date{Jun 2021}

\begin{document}

\maketitle

\begin{abstract} 
 Inference in clustering is paramount to uncovering inherent group structure in data. Clustering methods which assess statistical significance have recently drawn attention owing to their importance for the identification of patterns in high dimensional data with applications in many scientific fields. We present here a U-statistics based approach,  specially tailored for high-dimensional data, that clusters the data into \emph{three} groups while assessing the significance of such partitions.  Because our approach stands on the U-statistics based clustering framework of the methods in R package $uclust$, it inherits its characteristics being a non-parametric method relying on very few assumptions about the data, and thus can be applied to a wide range of dataset. 
  Furthermore our method aims to be a more powerful tool to find the best partitions of the data into three groups when that particular structure is present. In order to do so, we first propose an extension of the test U-statistic and develop its asymptotic theory. Additionally we propose a ternary non-nested significance clustering method. Our approach is tested through multiple simulations and found to have more statistical power than competing alternatives in all scenarios considered. Applications to peripheral blood mononuclear cells and to image recognition shows the versatility of our proposal, presenting a superior performance when compared with other approaches. 
  
\end{abstract}

\section{Introduction}


In clusters analysis the aim is to divide data into groups of similar items and  there are different ways to accomplish this task. A large number of algorithms based on different measures have been proposed and each different measure may lead to potentially different results (\cite{euan2019}). Clusters can be inherently present in the data like in phylogenetic analysis (\cite{rosenberg2002,chen2015}) or they can be built when clustering should take place regardless of whether innate cluster structure is present as in customer segmentation (\cite{motlagh2019, hennig2015}). In order to evaluate clustering methods, it is necessary to consider the context, the objectives of clustering and to have a suitable measure of dissimilarity (\cite{von2012}). A critical issue is how to discover inherent cluster structure in data, in other words, whether the clusters represent in fact an important feature or are simply the result of sample variation. This becomes even more challenging when considering the context of high dimensional data. We present here a U-statistics based approach that clusters the data in three groups while assessing the significance of such partitions. Our method is specially tailored for high-dimensional data and adaptable to different distance measures.

In a typical application of inference in clustering when the groups are already defined and there is no need for an algorithm or method to find them, the null hypothesis is that all groups are random samples from the same population (overall sample homogeneity). In the multivariate analysis of variance (M)ANOVA procedure,  when  presented in terms of a linear model, the homogeneity of groups stands for equality of means between all groups. Assumptions of independence and normality of the data, homoscedasticity of variance and homogeneity in group are required for exact (finite sample) inference. In addition, a large sample size, depending on the dimension of the data is generally necessary. For the context where there is no information about the existence of groups and the objective is to know if they exist and what they are, some approaches have been proposed for addressing the problem of assessing significance of partitions, or determining which clustering layers represent actual population structure and which are simple consequence of spurious random effects. To avoid resorting to heuristic criteria or the researcher's judgement  to define which partition levels should be assigned meaning these approaches proposes to assess statistical significance. However the success of these methods depends on the underlying cluster structure (\cite{adolfsson2019}).

Several approaches have been proposed to assess statistical significance in clustering, for example the procedure presented in \cite{mclachlan2004} which considers mixture models of distributions such as the Gaussian. A maximum likelihood approach is used by \cite{demidenko2018} to test no-clusters hypothesis. However,  when the data are high dimensional and have small sample sizes the problem becomes increasingly challenging, since it involves complete parametric estimation, usually requiring costly matrix inversions. The works of \cite{mcshane2002, helgeson2020} address this issue by using reduction of dimensionality of the data matrix and sparse covariance estimation. An approach inspired on the bootstrap strategy is proposed by \cite{shimodaira2004} which is implemented in the R package $pvclust$   (\cite{suzuki2006})  and used in phylogenetics to assess confidence in hierarchical clustering. \cite{Liu2008} proposes a statistical test to assess the significance of clustering the data into $K$ groups, specifically tailored to the high dimension low sample size (HDLSS) scenario,  that has been implemented in the R package $sigclust$. However, the implementation and applications consider only two groups. Additionally, \cite{kimes17} extend the method to assess significance in hierarchical clustering. However, this approach requires that the data comes from a single multivariate normal distribution, which can be an issue since rejection of the no cluster hypothesis may be a simple consequence of non-normal data.

Our work focuses specifically on the HDLSS setting and extends the works of \cite{cybis18,valk20} making it possible to simultaneously test the homogeneity of \emph{three} groups, one of which may have size \emph{one}. The test statistic to compare \emph{three} groups, where one of them may be an outlier, is a extension of the test statistic $B_n$ proposed by \cite{pinheiro09}. Here the hypotheses are similar to those of (M)ANOVA where the null is that the elements in the \emph{three} groups come from the same distribution (homogeneity, no-clusters) versus the alternative hypothesis that the data distribution (not necessarily normal) of at least one of the groups is different from the others. Asymptotic normality of the extended $B_n$  is obtained using U-statistics theory. An estimator for the variance of the extended $B_n$ is proposed. In addition, we have developed an algorithm ($uclust3$) that finds the best significant separation in \emph{three} groups. Simulation studies show that our proposal presents coherent results, such as control of Type I Error and the increased  Power to identify clusters as they become more separated. Furthermore, our comparative simulation study with other methods shows that in the case where there are exactly \emph{three} groups, the approach we are proposing has greater power, that is, greater ability to correctly identify \emph{three} clusters. More accurate results of $uclust3$ are found in an application to real image recognition data, corroborating the better performance of our approach observed in the simulations. Although we are using Euclidean distance and simulating data with normal distribution, these aspects are not essential to the validity of the method properties.

The steps to developing our three groups clustering method are outlined as follows. First, in Section \ref{sec::review} we review the U-statistics based theory of the  homogeneity test of \cite{cybis18} and present the U-statistics theory for \emph{three} groups. In Section \ref{subec:extBn} we present the extension of the $B_n$ statistics proposed by \cite{pinheiro09} to contemplate \emph{three} groups in which one may have  size one, in order to devise a clustering algorithm that can properly identify outlier elements. Additionally an investigation of theoretical properties that show its compatibility with the previous framework and asymptotic theory, is also presented. In Section \ref{subsec:varBn} we explore the variance aspects of the extended $B_n$ and propose an approach to estimate this variance. In Section \ref{sec:uclust3} we propose the $uclust3$ method which finds the statistically significant data partition that better separates the sample into three groups. The remainder of the paper focuses on evaluating the methodology through simulation studies, in Section \ref{sec::simulation}, and applications to real data in Section \ref{sec::application}. Finally, in Section \ref{sec:conclusions} we discuss the overall results.

\section{Methods}\label{sec:methods}

\subsection{U-Statistics based test for three group separation }\label{sec::review}

Let $\bf{X}=({\bf X}_{1},\dots,{\bf X}_{n})$ be a random sample of $n$
$L$-dimensional vectors  divided in three groups $G_1$, $G_2$ and $G_3$ of sample sizes $n_1$, $n_2$ and $n_3$, respectively, where $n=n_1+n_2+n_3$. In the $g$-th group, for $g \in \{1,2,3\}$, observations ${\bf X}^{(g)}_{1},\dots,{\bf X}^{(g)}_{n_g}$ are assumed to be independent and identically distributed with a $L$-variate distribution $F_g$.   Here, the distribution $F_g$  admits finite mean vector $\boldsymbol{\mu}_g$  and positive definite dispersion matrix $\boldsymbol{\Sigma}_g$ (not necessarily multi-normal). 
Following the approach of \cite{Sen2006} and \cite{pinheiro09}, we define the functional distance $\theta(F_g,F_{g'})$ as

\begin{equation}\label{eq:phi_kernel}
\theta(F_g,F_{g'})=\int\int\phi(x_1,x_2)dF_g(x_1)dF_{g'}(x_2), \quad  x_1,x_2\in \mathbb{R}^L,
\end{equation}
where $g,g'\in\{1,2,3\}$ and $\phi(\cdot,\cdot)$ is a symmetric kernel of order $2$. If we assume that $\theta(\cdot,\cdot)$ is a convex linear function of its marginal components, then we have
\begin{equation}\label{eq:uneq}
\theta(F_g, F_{g'}) \geq \frac12 \,\{\theta(F_g, F_g) +\theta(F_{g'}, F_{g'})\},
\end{equation}
for all distributions $F_g$ and $F_{g'}$, with equality holding whenever $\mu_g= \mu_{g'}$.

Note that the functional $\theta(\cdot,\cdot)$ can be used to define both distance within and between groups. It follows from U-statistics theory that an unbiased estimator of this functional for within group distance $\theta(F_g,F_g)$ is a generalized U-statistic \cite{Hoeffding1948}, with kernel $\phi( \cdot, \cdot)$, defined as 

\begin{equation}\label{eq:Unwithin}
 U_{n_g}^{(g)}=\dbinom{n_g}{2}^{-1}\sum_{1\leq i< j\leq n_g}\phi({\bf
X}^{(g)}_{i},{\bf X}^{(g)}_{j}),
\end{equation}
where $g\in\{1,2,3\}$. Analogously, the unbiased estimator for the between group functional distance $\theta(F_g,F_{g'})$ is  defined by  
\vspace{-2mm}
\begin{equation}\label{eq:Unbetween}
 U_{n_g,n_{g'}}^{(g,g')}=\frac{1}{n_g n_{g'}}\sum_{i=1}^{n_g}\sum_{j=1}^{n_{g'}}\phi({\bf
X}^{(g)}_{i},{\bf X}^{(g')}_{j}),
\end{equation}
where $g, g' \in\{1,2,3\}$ and $g \neq g'$.

The combined sample U-statistic is usually decomposed as 
\begin{eqnarray}\label{eq:statisticsUn}
U_{n} &=& \sum_{g=1}^{3} \frac{n_{g}}{n} U_{n_g}^{(g)}+\sum_{1 \leq g<g^{\prime} \leq 3} \frac{n_{g} n_{g^{\prime}}}{n(n-1)}\left\{2 U_{n_g, n_{g^{\prime}}}^{(g,g')}-U_{n_g}^{(g)}-U_{n_{g^{\prime}}}^{(g')}\right\}\nonumber\\
&=& W_n+B_n.
\end{eqnarray}
Decomposition \eqref{eq:statisticsUn} leads to the statistic $B_n$, which provides the focal point of our methodology,   
\begin{eqnarray}\label{eq:statisticsBn}
 B_{n}&=&
\sum_{1 \leq g<g^{\prime} \leq 3} \frac{n_{g} n_{g^{\prime}}}{n(n-1)}\left\{2 U_{n_g,n_{g^{\prime}}}^{(g,g')}-U_{n_g}^{(g)}-U_{n_{g^{\prime}}}^{(g')}\right\}.
\end{eqnarray}

Here $U_{n_g}^{(g)}$ for $g\in \{1,2,3\}$ are U-statistics associated to within group distances, as defined in \eqref{eq:Unwithin}, and $U_{n_g n_{g'}}^{(g,g')}$, $g\neq g'\in\{1,2,3\}$, are the U-statistics associated to between group distances as defined in \eqref{eq:Unbetween}. Note that the definition of $U_{n_g}^{(g)}$  require a minimum of 2 elements in the group. This imposes minimum group sizes $n_g\geq 2$, for $g\in \{1,2,3\}$ for proper definition of $B_n$.

The methodology proposed in \cite{cybis18} and \cite{valk20} considers a group homogeneity test which verifies whether two groups in fact constitute separated groups, or if they stem from the same distribution. In this work, for data arranged in three groups $G_1$, $G_2$ and $G_3$, the interest is in verifying whether the data are homogeneous or if there is at least one group statistically separated. Thus, the null hypothesis $H_0$ states that $F_1=F_2=F_3$, while the alternative $H_1$ states that there are $i\neq j$, $\in \{1,2,3\}$ where $F_i\neq F_j$. In cases where groups $G_1$, $G_2$ and $G_3$ have more than two elements, the asymptotic properties of $B_n$ are addressed in \cite{pinheiro09}.  The statistics $B_n$ is in the class of degenerate U-statistics for which asymptotic normality prevails and the convergence rates are $L$ and/or $\sqrt{n}$. Additionally, under the null, we have $\E(B_n)=0$ and under the alternative, $\E(B_n) > 0$. The null hypothesis is rejected for large values of standardized $B_n$, where the variance of $B_n$, under $H_0$, is obtained by a resampling procedure \cite{Sen2006}.

\subsection{The extension of test U-statistics for tree groups }\label{subec:extBn}

The homogeneity test proposed in \cite{cybis18} presents an essential concept for  our clustering algorithm. However, the group size restriction required by the definition of the U-statistic $B_n$ in \eqref{eq:statisticsBn} constrains this method to cases where all subgroups have sizes $n_i\geq 2$, $i=1,2,3$, and  consequently clustering methods  will fail in cases where the data has an outlier. In order to build a clustering algorithm that admits groups of size 1 we propose an extension of $B_n$. We can assume, without loss of generality, that only the group $G_1$ may have one element, and define

\begin{eqnarray} \label{eq:Bn3}
B_n=\left\{
\begin{array}{ll}
\frac{2n_2}{n(n-1)} \left( U_{1,n_2}^{(1,2)}-U_{n_2}^{(2)} \right) + \frac{2n_3}{n(n-1)} \left( U^{(1,3)}_{1,n_3} - U_{n_3}^{(3)} \right)   \\
\\
+ \frac{n_2n_3}{n(n-1)}\left(2U_{n_2,n_3}^{(2,3)}-U_{n_2}^{(2)}-U_{n_3}^{(3)}\right), 
 \hspace{5mm} \hbox{ if } n_1=1,\hbox{ and } n_2,n_3 >1 \\
\\
\displaystyle \sum_{1 \leq i < j \leq 3} \frac{n_i n_j}{n(n-1)} \left( 2U_{n_i,n_j}^{(i,j)} - U_{n_i}^{(i)}- U_{n_j}^{(j)} \right), \hspace{5mm} \hbox{if } n_1,n_2,n_3 >1. \\
\end{array}
\right.
\end{eqnarray}
where $U_{n_g,n_{g'}}^{(g,g')}$ and $U_{n_g}^{(g)}$ are defined, respectively, in \eqref{eq:Unbetween} and \eqref{eq:Unwithin}. 

This is a natural extension of $B_n$ considering data separation in three groups, when allowing for clusters of size 1. This extension coincides with that of expression \eqref{eq:statisticsBn} for group of sizes $n_1,n_2,n_3>1$, and thus all properties mentioned above are still valid for the new definition in that case.   We ascertain the validity of these asymptotic properties or analogous alternatives in the case of $n_1=1$.

Note that, when $G_1$ has size one, we can rewrite $B_n$ as

\begin{eqnarray*}
B_n &=& \frac{2n_2}{n(n-1)}U^{(1,2)}_{1,n_2}+\frac{2n_3}{n(n-1)}U^{(1,3)}_{1,n_3}+\frac{2n_2n_3}{n(n-1)}U^{(2,3)}_{n_2,n_3}\\
& & -\frac{n_2(2+n_3)}{n(n-1)}U^{(2)}_{n_2}-\frac{n_3(2+n_2)}{n(n-1)}U^{(3)}_{n_3} 
\end{eqnarray*}
where $U_{1,g}^{(1,g)}$ and $U_{n_g}^{(g)}$,  $g=2,3$ are as defined in \eqref{eq:Unbetween} and \eqref{eq:Unwithin}. If we consider the extension of $B_n$ in \eqref{eq:Bn3}, then we can write the combined sample U-statistics as

\begin{eqnarray*}
U_n &=& B_n+W_n^*.
\end{eqnarray*}
where $W_n^*$ is  an appropriate modification the term $W_n$.  Thus, $B_n$ still arises from the decomposition of the combined sample U-statistics into $B_n$ and a modified term  $W_n$. This extended definition allows us to build a U-test when a  group has size 1. We conveniently labeled the data in order to arrange the groups as follows. Let $G_1=\{{\bf X}_1\}$, $G_2=\{{\bf X}_{2},\dots,{\bf X}_{n_2+1}\}$ and $G_3=\{{\bf X}_{n_2+2},\dots,{\bf X}_{n}\}$, $n=1+n_2+n_3$. We still have $\E[B_n]=0$, under the null hypothesis of overall group homogeneity.  Additionally, if we make the assumption that  
\begin{eqnarray}\label{eq:assump1}
 \theta_{gg'}>\theta_g,
\end{eqnarray}
for $g\neq g' \, \in \{1,2,3\}$ where  $\theta_g= \operatorname{E}\left[ \phi ( X_g, X_g)\right]$ and $\theta_{gg'}=\operatorname{E}\left[ \phi ( X_g, X_{g'})\right]$, 
 then under alternative we have that $\E[B_n]>0$. Note that this assumption is usual and when \eqref{eq:assump1} is valid then equation (\ref{eq:uneq}) is always satisfied.

Asymptotic theory for the $B_n$ statistic for group sizes greater than 2 is developed in the work of \cite{pinheiro09}, where it is established that $B_n$ is a degenerate U-statistic and asymptotic normality is provided. The following  theorems demonstrate that the extended $B_n$ is a non degenerated U-statistics and establish the asymptotic distribution of the extended $B_n$ under $H_0$ for increasing dimension $L$ and sample size $n$, requiring regularity conditions akin to those of the $n_1,n_2,n_3>1$ case. The following Lemma is an important result required to demonstrate the asymptotic convergence of the test statistic.
\begin{lema}\label{lema:converg}
Let $\frac{X}{\delta _n} \xrightarrow{D} \operatorname{N}(0,1)$, $\delta _n=\operatorname{O}(1)$ and $\delta ^*_n= \operatorname{O}(1)$. Then, $\frac{X}{\delta ^*_n}\xrightarrow{D} \operatorname{N}(0,M)$ where $M=\lim_{n\rightarrow\infty} \left( \frac{\delta _n^2}{\delta _n^{*2}} \right)$.
\end{lema}

\begin{proof}
Note that\[\frac{X}{\delta ^*_n}\frac{\delta_n}{\delta_n}=\frac{\delta _n}{\delta ^*_n} \frac{X}{\delta _n} \xrightarrow{D} \operatorname{N}(0,\gamma),\]
where 
\[\gamma = \operatorname{Var}\left(\frac{\delta _n}{\delta ^*_n} \frac{X}{\delta_n} \right) \rightarrow \lim_{n\rightarrow\infty} \left( \frac{\delta _n}{\delta ^*_n} \right) ^2 = M.\]
\end{proof}

\begin{theorem}\label{Theorem1}\rm
Let ${\bf X}_1, {\bf X}_2,\dots,{\bf X}_n$ be a sequence of i.i.d. $L \times 1$ random vectors. Let $\phi(\cdot,\cdot)$ be a kernel of degree 2 satisfying $\E[\phi({\bf X}_1,{\bf X}_2)^2]<\infty$ and $\Var[\E(\phi({\bf X}_1,{\bf X}_2)|{\bf X}_1)]=\sigma_1^2>0$. Consider definition \eqref{eq:Bn3} for $B_n$ when $n_1=1$ and let $V_n = \operatorname{Var}(B_n)$, $\tau_n=(n/2)V_n^{1/2}$ and  $W=J_1+J_2-J_3-J_4$, where $\frac{\psi_1(X_1)}{\tau_n} \xrightarrow{D} J_1$, and $J_2, J_3 \hbox{ and } J_4$ are random variables with normal distribution. Then

\begin{equation}\label{eq:teo1convBn_n}
\frac{(n/2)B_n}{\tau_n} \xrightarrow{D} W \hbox{ as  n }\rightarrow \infty.
\end{equation}

\end{theorem}

\begin{proof}

We are interested in the distribution of $B_n$ with fixed $L$ and $n \rightarrow \infty$. Is is strightforward to show that $\tau_n=\frac{n}{2}\sqrt{\Var(B_n)}=O(1)$. From the Hoeffding decomposition of $B_n$ we have:

\begin{equation}
\frac{n}{2}B_n=W_1+W_2-W_3-W_4
\end{equation}
where

\begin{eqnarray}
W_1&=&\psi_1(X_1)-\frac{1}{n-1}\sum_{i=1}^{n_2} \psi_1(X_{2i}) - \frac{1}{n-1} \sum_{j=1}^{n_3} \psi_!(X_{3j})+ \nonumber\\
& &  +\frac{1}{n-1}\sum_{i=1}^{n_2}\psi_2(X_1,X_{2i}) +  \frac{1}{n-1}\sum_{j=1}^{n_3}\psi_2(X_1,X_{3j}) \\
W_2&=& \frac{1}{n-1}\sum_{i=1}^{n_2}\sum_{j=1}^{n_3} \psi_2(X_{2i},X_{3j}) \\
W_3&=& \frac{2+n_3}{(n-1)(n_2-1)}\sum_{1 \leq i < j \leq n_2} \psi_2(X_{2i},X_{2j}) \\
W_4&=&\frac{2+n2}{(n-1)(n_3-1)}\sum_{1 \leq i <j \leq n_3} \psi_2(X_{3i},X_{3j}).
\end{eqnarray}

 Under the null hypothesis  $\boldsymbol{X}_1,\hbox{ } \boldsymbol{X}_2 \hbox{ and } \boldsymbol{X}_3$ are identically distributed, thus $W_1$ can be expressed as

\begin{eqnarray}
W_1= \psi _1(X_1) - \frac{1}{n-1} \displaystyle \sum_{i=2}^{n} \psi _1(X_i)+ \frac{1}{n-1}\displaystyle \sum_{j=2}^n \psi _2(X_1,X_j).
\end{eqnarray}
By the Law of Large Numbers  (LLN) follows that

\begin{eqnarray}
\frac{1}{n-1}\displaystyle \sum_{i=2}^n \psi _1(X_i) \xrightarrow{P} \operatorname{E}[ \psi _1(X_1)]=0 \\
\frac{1}{n-1}\displaystyle \sum_{j=2}^n \psi _2(X_1,X_j) \xrightarrow{P} \operatorname{E} [\psi _2(X_1,X_2)]=0.
\end{eqnarray}

Thereby, \[W_1 \xrightarrow{P} \psi _1(X_1).\]

As $\frac{\psi _1(X_1)}{\tau _n} \xrightarrow{D} J_1$ and  $W_1 \xrightarrow{P} \psi _1(X_1)$, then, by Slutsky's theorem, $\frac{W_1}{\tau _n} \xrightarrow{D} J_1 \hbox{ as  n }\rightarrow \infty$.

From the Central Limit Theorem (TCL) we have

\begin{eqnarray}
\frac{W_2-E(W_2)}{\sqrt{\Var(W_2)}}= \frac{\frac{1}{n-1} \displaystyle\sum_{i=1}^{n_2}\sum_{j=1}^{n_3}\psi_2(X_{2i},X_{3j})}{\sqrt{\frac{n_2n_3}{{(n-1)^2}}\tau_2^2}} \xrightarrow{D} N(0,1) \hbox{ as  n }\rightarrow \infty.
\end{eqnarray}

Observe that $\sqrt{\frac{n_2n_3}{{(n-1)^2}}\tau_2^2} = O(1)$ and  $\tau _n = O(1)$. Then by Lemma \ref{lema:converg} follows that

\begin{eqnarray}
\frac{W_2}{\tau _n} \xrightarrow{D} J_2 \sim \operatorname{N}(0,M_2)\hbox{, where }M_2 = \lim_{n\rightarrow\infty} \left( \frac{\frac{n_2n_3}{{(n-1)^2}}\tau_2^2}{\tau ^2_n}\right).
\end{eqnarray}

Similarly,

\begin{eqnarray}
\frac{W_3-\operatorname{E}(W_3)}{\sqrt{\Var(W_3)}}=\frac{\frac{(2+n_3)}{(n-1)(n_2-1)} \displaystyle \sum_{1 \leq i < j \leq n_2} \psi _2(X_{2i},X_{2j})}{\sqrt{\frac{(2+n_3)^2 n_2 \tau_2^2}{2(n-1)^2(n^2-1)}}} \xrightarrow{D}N(0,1) \hbox{ as  n }\rightarrow \infty.
\end{eqnarray}

Other properties are that $\sqrt{\frac{(2+n_3)^2 n_2 \tau_2^2}{2(n-1)^2(n^2-1)}}=O(1)$ and $\tau _n = O(1)$, then by the Lemma \ref{lema:converg} 

\begin{eqnarray}
\frac{W_3}{\tau_n} \xrightarrow{D} J_3 \sim \operatorname{N}(0,M_3) \hbox{, where } M_3 = \lim_{n\rightarrow\infty} \left( \frac{\frac{(2+n_3)^2 n_2 \tau_2^2}{2(n-1)^2(n^2-1)}}{\tau ^2_n} \right).
\end{eqnarray}

Analogously,

\begin{eqnarray}
\frac{W_4 - \operatorname{E}(W_4)}{\sqrt{\Var}(W_4)} = \frac{\frac{(2+n_2)}{(n-1)(n_3-1)}\displaystyle \sum_{1 \leq i < j \leq n_3} \psi _2(X_{3i},X_{3j})}{\sqrt{\frac{(2+n_2)^2n_3\tau_2^2}{2(n-1)^2(n_3-1)}}} \xrightarrow{D} N(0,1)\hbox{ as  n }\rightarrow \infty.
\end{eqnarray}

Once more, $\sqrt{\frac{(2+n_2)^2n_3\tau_2^2}{2(n-1)^2(n_3-1)}}=O(1)$ and $\tau_n = O(1)$, then

\begin{eqnarray}
\frac{W_4}{\tau _n} \xrightarrow{D} J_4 \sim \operatorname{N}(0,M_4) \hbox{, where } M_4=\lim_{n\rightarrow\infty} \left( \frac{\frac{(2+n_2)^2n_3\tau_2^2}{2(n-1)^2(n_3-1)}}{\tau ^2_n} \right).
\end{eqnarray}

Thus, applying Slutsky's theorem we have 

\begin{eqnarray}
\frac{(n/2)B_n}{\tau _n} &=& \frac{(n/2)B_n}{(n/2)V_n^{1/2}} =  \frac{B_n}{\sqrt{\Var(B_n)}} \nonumber \\  
&=& \frac{W_1+W_2 - W_3 - W_4}{\tau _n} \xrightarrow{D} J_1 + J_2 - J_3 - J_4 \hbox{ as  n }\rightarrow \infty. 
\end{eqnarray}

\end{proof}

This result shows that the test statistic asymptotically converges in $n$ to a non-degenerate random variable whose limit distribution depends on the choice of kernel $\phi(\cdot,\cdot)$.

\begin{theorem}\label{Theorem2}\rm
Let ${\bf X}_1, {\bf X}_2,\dots,{\bf X}_n$ be a sequence of i.i.d. $L \times 1$ random vectors. Let $\phi(\cdot,\cdot)$ be a kernel of degree 2 such that
\begin{equation}\label{phistar}
\phi({\bf X}_i,{\bf X}_j)=\frac{1}{L}\sum_{l=1}^{L}\phi^*(X_{il},X_{jl})
\end{equation}
for some kernel $\phi^*(\cdot,\cdot): \R^2\rightarrow \R$, where $X_{il}$ is the $l$-th entry of ${\bf X}_i$.   Define $\phi_1^*(x_{il})=\E[\phi^*(X_{il},X_{jl})|X_{il}=x_{il}]$ and suppose $\Var(\phi_1^*(X_{il}))>0$ and $\Var(\phi^*(X_{il},X_{jl}))<\infty $. Let $B_n$ be defined by \eqref{eq:Bn3} for the case where $n_1=1$, and assume that all conditions in Theorem \ref{Theorem1} hold.   Suppose also that
\begin{equation}\label{covcond1}
\sum_{1\leq l<m\leq n}\E[\phi^*(X_{il},X_{jl})\phi^*(X_{im},X_{jm})]=O(L) 
\end{equation}
and 
\begin{equation}\label{covcond2}
\sum_{1\leq l<m\leq L}\E[\phi_1^*(X_{il})\phi_1^*(X_{jm})]=O(L).
\end{equation}
Then 
\begin{equation}\label{convBn_nL}
\frac{B_n}{\sqrt{\Var(B_n)}}\xrightarrow{D} N(0,1) \quad \mbox{ as } \quad  L\rightarrow \infty.
\end{equation}
\end{theorem}

\begin{proof}
We start writing $\psi _1(X_i)$ and $\psi _2(X_i,X_j)$ as a function of $\phi ^*_1(\cdot)$ and $\phi^*_2(\cdot,\cdot)$. Note that

\begin{eqnarray}
\psi_{1} \left(\mathbf{X}_{i}\right)&=&\frac{1}{L} \sum_{l=1}^{L} \psi_{1}^{*}\left(X_{i l}\right) \\
\psi_{2}\left(\mathbf{X}_{i}, \mathbf{X}_{j}\right)&=&\frac{1}{L} \sum_{l=1}^{L} \phi ^* (X_{il},,X_{jl}) - \frac{1}{L} \sum_{l=1}^{L} \psi ^*_1(X_{il}) \nonumber \\ 
& &-  \frac{1}{L} \sum_{l=1}^{L} \psi ^*_1(X_{jl}) - \theta
\end{eqnarray}
where

\begin{eqnarray}
\psi ^*_1(X_{il}) &=& \phi ^*_1(X_{il})- \theta \\
\phi ^*_1(x_{il})&=& \E[\phi ^* (X_{il},X_{jl}) \mid X_{il} = x_{il}] \\
\phi ^*_2(x_{il},x_{jl}) &=& \E[ \phi ^*(X_{il},X_{jl}) \mid X_{il}=x_{il}, X_{jl} =x_{jl}].
\end{eqnarray}

We can write $\psi_1(\cdot)$ as

\begin{eqnarray}
\psi _1(\boldmath{X}_i)&=& \frac{1}{L} \sum_{l=1}^{L} \left[ \phi ^*_1(X_{ij})-\theta \right],
\end{eqnarray}
or

\begin{eqnarray}
\psi _1(\boldmath{X}_i) &=& \frac{1}{L} \sum_{l=1}^{L} \psi _1^*(X_{ij}).
\end{eqnarray}

Thus the variance of $\psi_1(\cdot)$ is given by

\begin{eqnarray}
\Var(\psi _1(\boldmath{X}_i)) &=& \Var \left[ \frac{1}{L} \sum_{l=1}^L \psi _1^*(X_{il}) \right].
\end{eqnarray}
By (\ref{covcond1}) we have that

\begin{eqnarray}
\Var\left(\psi_{1}\left(\mathbf{X}_{i}\right)\right) &=&\frac{1}{L^{2}}\left\{\sum_{l=1}^{L} \Var\left[\psi_{1}^{*}\left(X_{i l}\right)\right]\right.\nonumber\\
&&\left. +2 \sum_{1 \leq l<m \leq L} \Cov\left(\psi_{1}^{*}\left(X_{i l}\right), \psi_{1}^{*}\left(X_{im}\right)\right)\right\} \nonumber\\
&=& O\left(L^{-1}\right)
\end{eqnarray}
and by (\ref{covcond2}) the variance of $\psi _2(,)$ is

\begin{equation}\begin{aligned}
\Var\left(\psi_{2}\left(\mathbf{X}_{i}, \mathbf{X}_{j}\right)\right)=& \frac{1}{L^{2}}\left\{\sum_{l=1}^{L} \Var\left(\phi^{*}\left(X_{i l}, X_{j l}\right)\right)+ \right. \\
& + 2 \sum_{1 \leq l<m \leq L} \Cov\left(\phi^{*}\left(X_{i l}, X_{j l}\right), \phi^{*}\left(X_{im}, X_{jm}\right)\right) \\
& \left.+2 \Var\left(\frac{1}{L} \sum_{l=1}^{L} \psi_{1}^{*}\left(X_{i l}\right)\right)\right\} \\
=& O\left(L^{-1}\right).
\end{aligned}
\end{equation}

Thus, for fixed $n$ and for $L \rightarrow \infty$ it follows that

\begin{eqnarray}
\frac{B_n}{\sqrt{\operatorname{Var}(B_n)}}=V_{n}^{-1 / 2} B_{n} \xrightarrow{D}  N(0,1).  
\end{eqnarray}

\end{proof}

This result is fundamental to our inference procedure for clustering in the HDLSS context.

\subsection{Variance of  $B_n$}\label{subsec:varBn}

In the $utest$ the estimation of $B_n$'s variance under $H_0$ plays an essential role in hypothesis testing (see \cite{cybis18} ). As shown below, even under $H_0$, the variance of $B_n$ depends on the particular group configuration under consideration. For the homogeneity test of Section \ref{sec:uclust3}, we must evaluate this variance for the many group configurations visited in an optimization algorithm. This variance estimation is performed through a resampling procedure, however it becomes computationally expensive to perform one resampling procedure for each individual group size configuration. To circumvent this issue, \cite{cybis18} propose a reweighting scheme taking advantage of analytic calculations for the variance for the case $K=2$ groups. They are able to compute all variances from a single resampling procedure. In this section we extend their argument to the case of $K=3$ groups. 

In this Section we provide an estimator for the variance of $B_n$ under $H_0$ based on U-statistics properties of $B_n$. For cases where all groups have more than two elements, the Hoeffding decomposition of $B_n$ can be found in  \cite{pinheiro09} which is given by

\begin{eqnarray}\label{eq:decompPinheiro}
B_n &=& \left( \frac{2}{n(n-1)} \right) \sum_{1 \leq i < j \leq n} \eta _{nij} \psi_2(X_i,X_j),
\end{eqnarray}
where $\psi_2(\dot,\dot)$ is the second order term of the Hoeffding decomposition of $B_n$ and

\begin{eqnarray}\label{eq:etaPinheiro}
\eta_{nij}  &=& \left\{
\begin{array}{ll}
1, & \hbox{ if } i \hbox{ and } j \hbox{ are from different groups} \\
\\
-\frac{(n-n_g)}{n_g-1}, & \hbox{ if } i \hbox{ and } j \hbox{ are from the same group } g. 
\end{array}
\right.
\end{eqnarray}

Thereby,

\begin{eqnarray}\label{eq:vareta}
\operatorname{Var}(B_n)= \left( \frac{2}{n(n-1)} \right) ^2 \tau _2^2 \sum_{1 \leq i<j \leq n} \eta^2_{nij}.
\end{eqnarray}
where $\tau _2 ^2=\Var(\psi _2(X_1,X_2))$. From \cite{pinheiro09} we also know that

\begin{eqnarray}\label{eq:sumeta}
\sum_{1 \leq i<j \leq n} \eta_{n i j}^{2}={n \choose 2}(G-1)\left\{1+\frac{1}{n} \sum_{g=1}^{G} \frac{n-n_{g}}{\left(n_{g}-1\right)(G-1)}\right\}.
\end{eqnarray}

\noindent For the case in which we have three groups, $G_1$, $G_2$ and $G_3$, with sizes $n_1$, $n_2$ and $n_3$, respectively, where $n_1 + n_2 + n_3 = n$, it can be rewritten as

\begin{eqnarray}\label{eq:Cn}
\label{somaeta}
C_n(n_1,n_2)=\sum_{1 \leq i<j \leq n} \eta_{n i j}^{2}=2{n \choose 2}\left\{1+\frac{1}{n} \sum_{g=1}^{3} \frac{n-n_{g}}{2\left(n_{g}-1\right)}\right\},
\end{eqnarray} 
and therefore

\begin{eqnarray}\label{eq:varbncn}
\Var(B_n) &=& \left( \frac{2}{n(n-1)} \right) ^2 \tau _2^2 C_n(n_1,n_2)=V_{n_1,n_2}.
\end{eqnarray}

\noindent Note that only $\tau_2^2$ depends on the probability distribution of the data. Given three groups of sizes $n_1$, $n_2$ and $n_3$, the variance of $B_n$ for this configuration is estimated through a resampling  procedure. For optimization purposes, it is not interesting to perform a resampling procedure for each group configuration, so the idea is to use (the relation) expression  \eqref{eq:varbncn} to estimate $B_n$'s variance for any group configuration from a single resampling procedure. Let
$G_1^*$, $G_2^*$ and $G_3^*$, with sizes $n_1^*$, $n_2^*$ and $n_3^*$, respectively, where $n_1^* + n_2^* + n_3^* = n$, be an other group configuration for the same data set. From \eqref{eq:varbncn} it follows that

\begin{eqnarray}\label{eq:vv*}
V_{n_1^*,n_2^*}&=& \frac{C_n(n_1^*,n_2^*)}{C_n(n_1,n_2)}V_{n_1,n_2}.
\end{eqnarray}

Thus estimating $V_{n_1,n_2}$ through a resampling procedure is sufficient to estimate the variance of $B_n$ for any other group configuration. Although the variance of $B_n$ is estimated under $H_0$, we note that the choice of $n_1$ and $n_2$ may be important to reduce the bias of the variance estimator. 
To understand the $C_n(\cdot,\cdot)$ function's behavior we plot \eqref{eq:Cn} 
assuming that $n_1, n_2, n_3\geq2$ and $n=n_1 + n_2 + n_3$. 
As $\tau_2^2$ does not depend on group sizes, the behavior of $C_n(\cdot,\cdot)$ governs the behavior of $B_n$'s variance and Figure \ref{fig:Cn} shows that smaller values are obtained when groups have balanced sizes, while larger values of $C_n(\cdot,\cdot)$ are obtained when group sizes are unbalanced.


\begin{figure}[htb!]
\centering
\includegraphics[scale=0.4]{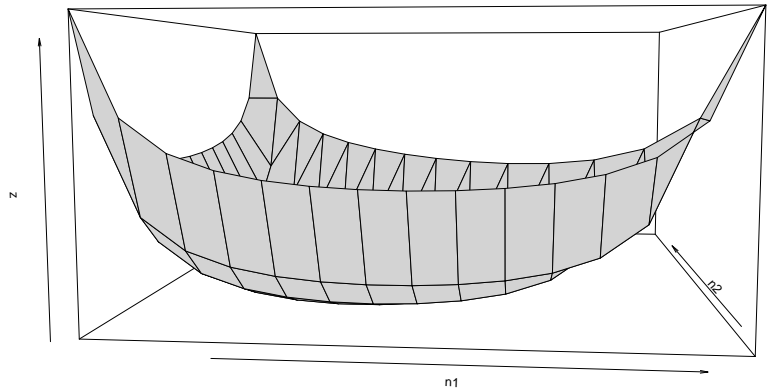}
\vspace{-1cm}
\caption{ $C_n(\cdot,\cdot)$ function behavior for $n_1, n_2, n_3\geq2$ and $n=n_1 + n_2 + n_3$.}
\label{fig:Cn}
\end{figure}

\pagebreak
\newpage

\subsubsection{Variance of the extended $B_n$}

We propose an extended statistic $B_n$ in \eqref{eq:Bn3} to accommodate cases in which the data set is divided into three groups, one of which has size one. For inference purposes it is essential establish a strategy to estimate the variance of the extended $B_n$. Through the Hoeffding decomposition of \eqref{eq:Bn3} (see Suplementary Material) we have that the variance of the extended $B_n$ is

\begin{eqnarray}\label{eq:varBn1}
\Var(B_n)= \zeta _1(n) \tau_1^2 + \zeta _2(n,n_2) \tau _2^2,
\end{eqnarray}
where $\tau _1^2 = \operatorname{Var}(\psi _1(X_1))$ and $\tau _2^2 = \operatorname{Var}( \psi _2(X_1,X_2))$ are, respectively, the variance of the first and the  second order terms of the Hoeffding decomposition,

\begin{eqnarray}\label{eq:zeta12}
\zeta _1(n)&=&\frac{4}{n(n-1)}, \nonumber\\
\zeta _2(n,n_2)&=& \frac{4}{n^2(n-1)} + \frac{4n_2n_3}{n^2(n-1)^2}+ \frac{2n_2(2+n_3)^2}{n^2(n_2-1)(n-1)^2}\\
& & +\frac{2n_3(2+n_2)^2}{n^2(n_3-1)(n-1)^2}, \nonumber
\end{eqnarray}
$n_1=1$, and $n_3=n-n_2-1$. Note that in expression \eqref{eq:varBn1} the terms $\tau_1^2$ and $\tau_2^2$ depend on the probability distribution of the data,  $\zeta _1(\cdot)$ depends only on $n$ and  $\zeta_2(\cdot,\cdot)$ depends on $n$ and $n_2$ since $n_3=n-n_2-1$. Thus for another group configuration keeping one of the groups with size one, the only change occurs at $n_2$, say $n_2^*$. For this new group configuration, the extended $B_n$ variance is given by

\begin{eqnarray}\label{eq:varBn1*}
\Var(B_n) = \zeta _1(n) \tau _1^2 + \zeta _2(n,n_2^*) \tau _2^2. 
\end{eqnarray}

Again, the choice of $n_2$ may affect the variance of the estimator. Denoting \eqref{eq:varBn1} by $V_{n_2}$ and \eqref{eq:varBn1*} by $V_{n_2^*}$, we have
 from simple algebra that

\begin{eqnarray}\label{eq:varBn12*}
V_{n_2^*} &=& V_{n_2} + [ \zeta _2(n,n_2^*) - \zeta _2(n,n_2)] \tau _2^2. 
\end{eqnarray}

For a given $n_2$ we can estimate $V_{n_2}$ from a resampling procedure. 
Additionally, an estimate for $\tau _2^2$ can be obtained from the strategy employed to estimate the variance of $B_n$ without outlier through expression  \eqref{eq:varbncn} as

\begin{eqnarray}\label{eq:tauest}
\widehat{\tau}_2^2&=& \frac{\widehat{V}_{n_1,n_2}}{C(n_1,n_2)\left( \frac{2}{n(n-1)}\right)^2}.
\end{eqnarray}

Thus we have a procedure to estimate the extended $B_n$´s variance for any group configuration from only two independent resampling procedures, through expression

\begin{eqnarray}\label{eq:varBn1est}
\widehat{V}_{n_2^*} &=& \widehat{V}_{n_2} + [ \zeta _2(n,n_2^*) - \zeta _2(n,n_2)] \widehat{\tau}_2^2,  
\end{eqnarray}

where $\widehat{\tau}_2^2$ is obtained from the resampling employed to estimate the variance of $B_n$ without outlier and $\widehat{V}_{n_2}$ is obtained from an additional resampling specific to $n_1=1$ case. Thus, taking into account the resampling procedure performed to estimate the variance of  $B_n$ when the groups are larger than two and, with one more resampling procedure for the size one group, we have an estimator for extended $B_n$'s variance.

In Figure \ref{fig:zeta2} we have the behavior of $\zeta_2(n,n_2)$ as a function of $n_2$.

\begin{figure}[h!]
\centering
\includegraphics[scale=0.25]{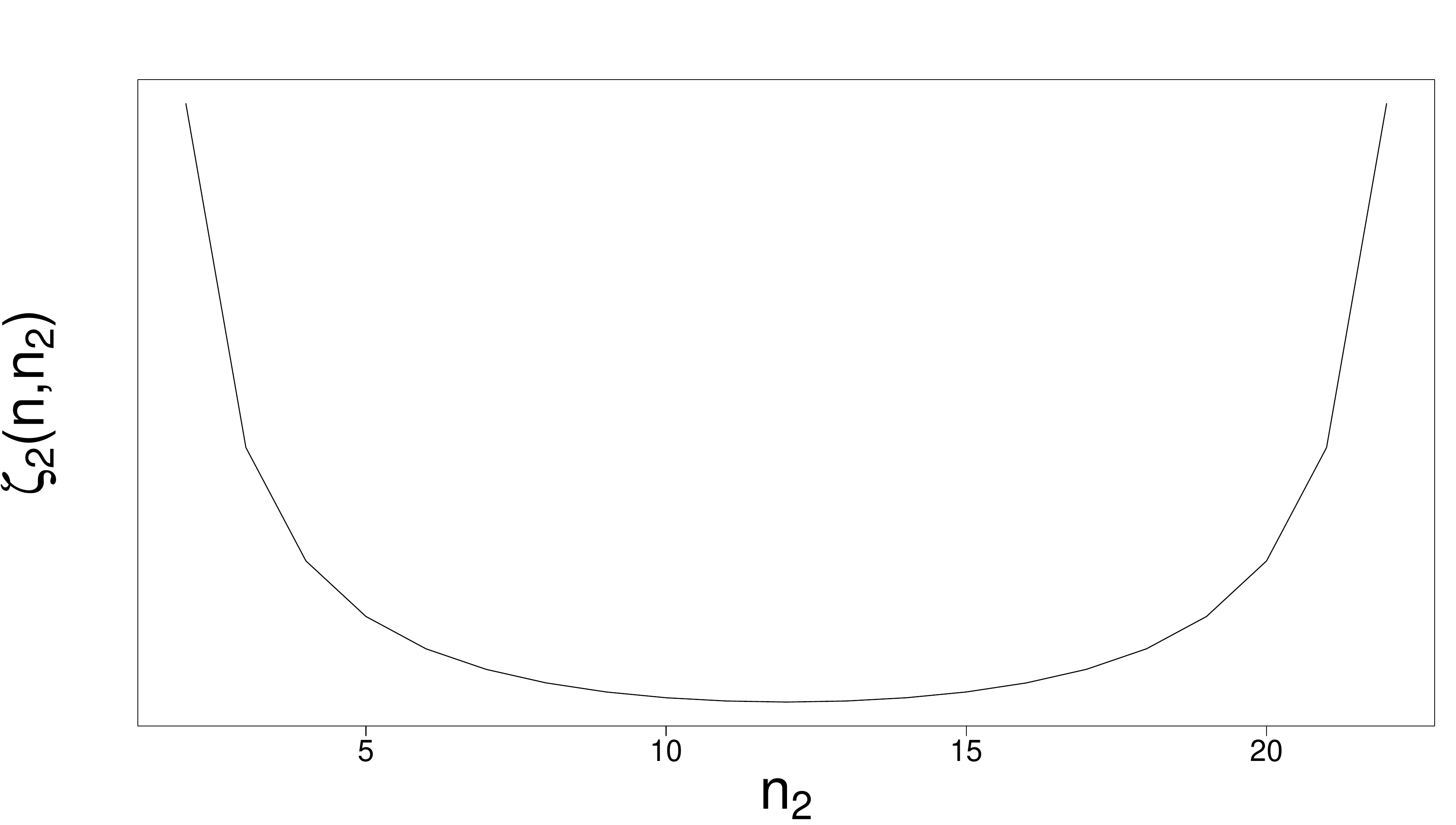}
\caption{Behavior of function $\zeta_2(n,n_2)$  for a given $n$, with $n_1=1$ and $n=1+n_2+n_3$.}
\label{fig:zeta2}
\end{figure}
These results are fundamental for the development of feasible algorithms that find significant clusters which is computationally challenging problem.



\pagebreak


\section{Homogeneity test for three groups}\label{sec:uclust3}

 Assessment of group homogeneity is a great challenge for standard statistics, especially in the HDLSS context. The $uclust$ algorithm presented in \cite{cybis18} and \cite{valk20} is effective to assess overall group homogeneity by verifying whether there exists some significant partition of the data in two groups. Here we are proposing an extension of the $uclust$ algorithm for data partitions in three groups $G_1$, $G_2$ and $G_3$.  
A combinatorial procedure like the one proposed by \cite{valk12} in which a $utest$ is applied for each possible partition of all group elements into three subgroups has serious computational restrictions due to the exponential increase in the number of tests that need to be performed.


\subsection{Total of combinations}\label{Asec:comb}

In order to develop the homogeneity test we require the number of different group configurations that can be formed by separating $n$ elements, $ x_1, x_2, \dots, x_n $ into three groups, $G_1,G_2$ e $G_3$. Follows from \cite{valk12} that the number of combination of $n$ elements into two groups is 
$p(n) = 2^{n-1}-n-1$.
Then if we divide $n$ elements into three groups where one of them  has size 1, it follows that the number of combinations is 
 
\begin{eqnarray}
\label{Aeq:delta3}
\delta _3(n)&=& (2^{n-2}-n)n.
\end{eqnarray}

Now we focus on the case where all groups have more than one element. We can fix, without loss of generality, $x_1$ as an element that belongs to the first group, $ G_1 $. Thus, we still have $ n-1 $ elements to be distributed among the three groups. Since we cannot have a unitary group, we need at least one more point for the first group. This group can have up to $ n-4 $ observations, since the remaining sets must necessarily have two elements each. Thus, we then have the following number of possible first sets

\begin{eqnarray*}
{n-1 \choose 1} + {n-1 \choose 2} + \cdots + {n-1 \choose n-5}.
\end{eqnarray*}
For the remaining elements that need to be divided into two clusters, just divide them into two groups with at least $2$ elements in each using the function $p(\cdot)$. Combining these results, we have a number of different configurations of non-unitary groups when we separate $n$ elements into $3$ groups given by

\begin{eqnarray}
\label{Aeq:s3}
S_3(n)&=& {n-1 \choose 1}p(n-2) + {n-1 \choose 2}p(n-3) + \cdots + {n-1 \choose n-5}p(4) \nonumber \\
&=& \sum_{k=1}^{n-5}{n-1 \choose k}p(n-k-1).
\end{eqnarray}

We can still rewrite this equation on a recurring basis. Note that if we already know how many configurations of groups we have with $n$ non-unitary elements, and how many configurations with a unitary group, then it is possible to calculate $S_3(n + 1)$ as

\begin{eqnarray}
\label{Aeq:s3rec}
S_3(n+1)&=& 3 S_3(n) + \delta _3(n).
\end{eqnarray}

With such equations we can rewrite $S_3(n) $ as

\begin{eqnarray}\label{Aeq:reco}
S_3(n)&=& \frac{233 (3^{n-6}) +1+n+n^2-(2+n)2^{n-1}}{2}.
\end{eqnarray}

Thus, the number of different group configurations where at most one of them  has size one is given by

\begin{eqnarray}
\label{eq:gamma3}
\gamma _3(n)&=& \frac{233 (3^{n-6}) +1+n+n^2-(2+n)2^{n-1}}{2} + \delta _3(n)\nonumber\\
 &=& \frac{233 (3^{n-6}) +1+n-n^2-2^{n} }{2} .
\end{eqnarray}
%
%
%

%
\noindent which becomes computationally onerous, especially for large sample size $n$.  To address this issue, we proceed similarly to  \cite{cybis18} proposing an optimization procedure to assess group homogeneity by finding the group configuration $G_1$, $G_2$ and $G_3$ that maximizes the objective function
\begin{equation}\label{eq:Bnst}
 f(G_1,G_2,G_3)=\frac{B_n}{\sqrt{\Var(B_n)}}.
\end{equation}
 By maximizing the standardized $B_n$ we must apply only one test. If this three group partition is found significant, then there is at least one subgroup that is significantly different from the others.  However, if $H_0$ is not rejected for this partition, then all other three group partitions will also be non-significant, and the whole data will be considered homogeneous. While only the group configuration with maximum standardized $B_n$ is tested we have to consider the distribution of $B_n$'s maximum under $H_0$. Making the untrue, but useful, simplifying assumption that the $B_n$'s are independent for different group configurations, the asymptotic cumulative distribution function of the maximum standardized $B_n$ is given by
\begin{equation}\label{eq:testmaxBn}
 F_{\mbox{max}}(x)= \Prob\left(\mbox{max}\left(\frac{B_n}{\sqrt{\Var(B_n)}}\right)<x\right)=\Phi(x)^{n^*},
\end{equation}
where $n^*=\gamma_3(n)$, for $\gamma_3(n)$ defined in \eqref{eq:gamma3} and $\Phi(\cdot)^{n^*}$ is the standard normal cumulative distribution function at the power $n^*$. For $F_{\mbox{max}}(x)>1-\alpha$, we reject the null hypothesis of overall group homogeneity with $\alpha$ significance level. 

 The number of tests increases rapidly, even for moderate sample size due to the combinatorial nature of our approach. The maximum distribution in \eqref{eq:testmaxBn} adequately accounts for multiple testing for reasonably small values of $n^*$. However, this approach has some shortcomings since $n^*$ rapidly increases.  Proceeding similarly to  \cite{valk20} and considering the simplifying assumption that the $B_n$'s are independent, we use extreme value theory and model it as Gumbel. However, the Gumbel approximation is only valid for very large values of $n^*$. Thus, for small $n$ we employ the standard max distribution of \eqref{eq:testmaxBn}, and when $n^*\geq 2^{28}$ the Gumbel distribution.

\subsection{The clustering method $uclust3$}\label{subsec:uclust3}

Our homogeneity test in the Section \ref{sec:uclust3} is a method that finds the configuration of \emph{three} subgroups that maximizes the standardized $B_n$.  This is appropriate for the context, since if the homogeneity test accepts the null for this partition, then it would also be accepted for all other partitions. However, the standardized $B_n$ might not be the best criteria to choose between competing partitions when more than one significant group separation exists.  This issue is addressed in \cite{cybis18} and arises from the fact that the variance of $B_n$ has different magnitudes depending on subgroup sizes $n_1$ and $n_2$ (expression \eqref{eq:Cn} dictates the relationship between variances, which is shown in Figure \ref{fig:Cn}). Consequently, this criteria favours partitions with group sizes of smaller variance, namely  $n_1,n_2\approx n/3$. We note that  the magnitude of the variance is quite different when we have a size one group, being much smaller in that case. Again if we use the standardized $B_n$ statistic as a criterion, we will have an effect of choosing groups of size one over the configurations of groups that present greater variance according to the Figure \ref{fig:Cn}.

Considering this issue, we proceed similarly to \cite{valk20} starting by testing overall group homogeneity which is based on maximum of standardized $B_n$. If the dataset is not homogeneous we adopt instead the maximum $B_n$ as the criteria for finding the configuration that better divides the sample into \emph{three} groups. Thus our significance clustering algorithm $uclust3$ will find the partition with maximum $B_n$ among the universe of all significant partitions in \emph{three} groups. 
 This is sufficient to ensure that the chosen configuration is statistically significant. 
 However, it is not efficient to find all arrangements of the data in \emph{three} groups that are statistically significant. Furthermore, we cannot simply test the clusters that maximizes $B_n$ since there are non-homogeneous samples for which this maximal partition is not significant.

 Based on these characteristics of the $B_n$ we propose a restricted search algorithm, which is based on the behaviors of the $B_n$'s variances (see Figure \ref{fig:Cn}). It starts from the group configuration that maximizes $B_n$ and if that partition is not significant, it searches for partitions whose $B_n$'s variances are smaller than the previous one. This is suitable since only for smaller variances, standardized $B_n$ can be significant. The equation  \eqref{eq:vv*} is used to avoid a new resampling 
 procedure to estimate the $B_n$'s variance. As there is a difference in the magnitudes of the $B_n$'s variances (see Figures \ref{fig:Cn} and \ref{fig:zeta2})  this algorithm treats separately the cases when we have a group of size one and the cases with no outlier. The detailed algorithm can be found in Section S3 of the supplementary materials.


\section{Simulation Studies}\label{sec::simulation}

In this section we present simulation studies in order to evaluate some aspects of our proposed methodology. For that we simulate canonical data and use the euclidean distance on our studies, but those are not mandatory for our methods. As presented in Section \ref{subsec:varBn}, $B_n$'s variance has a behavior that depends on the groups sizes. Moreover when we have a size one group, the order of magnitude of the  $B_n$'s variance is quite different when compared to cases in which groups sizes are larger than one. For this reason, our simulations studies typically have a configuration in which a group has size 1 and another configuration in which all groups have more than one element.  Figures \ref{fig:Cn} and \ref{fig:zeta2} show that $B_n$'s variance is smaller at a central group configuration, where the three groups have approximately the same number of elements. Conversely, the variance is greater for extreme group configurations, in which one of the groups has only two elements and the other has $n/2$ elements (or $n-1-n_2$ elements for cases where we have a group of size one). Naturally, the third group's size is defined as $n_3=n-n_1-n_2$. These scenarios are explored in our simulation studies. 

In the Section \ref{subsec:ut} we evaluate the empirical size and power of the proposed {\bf $utest$ for homogeneity of three groups}.  Section \ref{subsec:ht} present a simulation study to evaluate the empirical properties of the homogeneity test $uclust3$. The ability to find correct clusters of $uclust3$ and $kmeans$ clustering are compared in Section \ref{subsec:cc}.

\subsection{Simulations for the $utest$ }\label{subsec:ut}

We present here a simulation study to evaluate the empirical performance of the $utest$ for three groups. We simulate data from independent normally distributed (i.i.d.) samples  divided in three groups $G_1$, $G_2$ and $G_3$. The elements of the $L$ dimensional vectors in $G_1$ are generated from i.i.d. normal with mean $m_1=0$ and standard deviation equal to one. The vectors in $G_2$ and $G_3$ have the same properties with mean $m_2$ and $m_3$, respectively. 
In order to allow a graphical representation of the power of the test which is the proportion of rejection considering a significance level $\alpha$ (the power curves), the groups were symmetrically separated and on the x-axis the difference $m_2 - m_1$ is reported. The difference $m_3-m_2 = m_2 - m_1$. The sample size  $n$ takes values in $\{10,20,50\}$. Figure \ref{fig:pw1}  presents power curves of  the $utest$ for \emph{three} groups with separation degree $m_2 - m_1$,  where the vectors have  dimension $L=1000$ (gray) and $L=2000$ (black) and we have 100 replications of each scenario. Furthermore group $G_1$ has size one and group $G_2$ was set to have size  $n_2=\lfloor n/3\rfloor$, where  $\lfloor x \rfloor$ means the integer part of $x$.  Naturally the  third group's size is defined as $n_3=n-1-n_2$. The significance level used to determine whether the test rejects the null hypothesis that the elements in $G_1$, $G_2$ and $G_3$ have the same distribution was $\alpha = 0.05$.

\begin{figure}[h!]
\begin{center}
\includegraphics[scale=0.25]{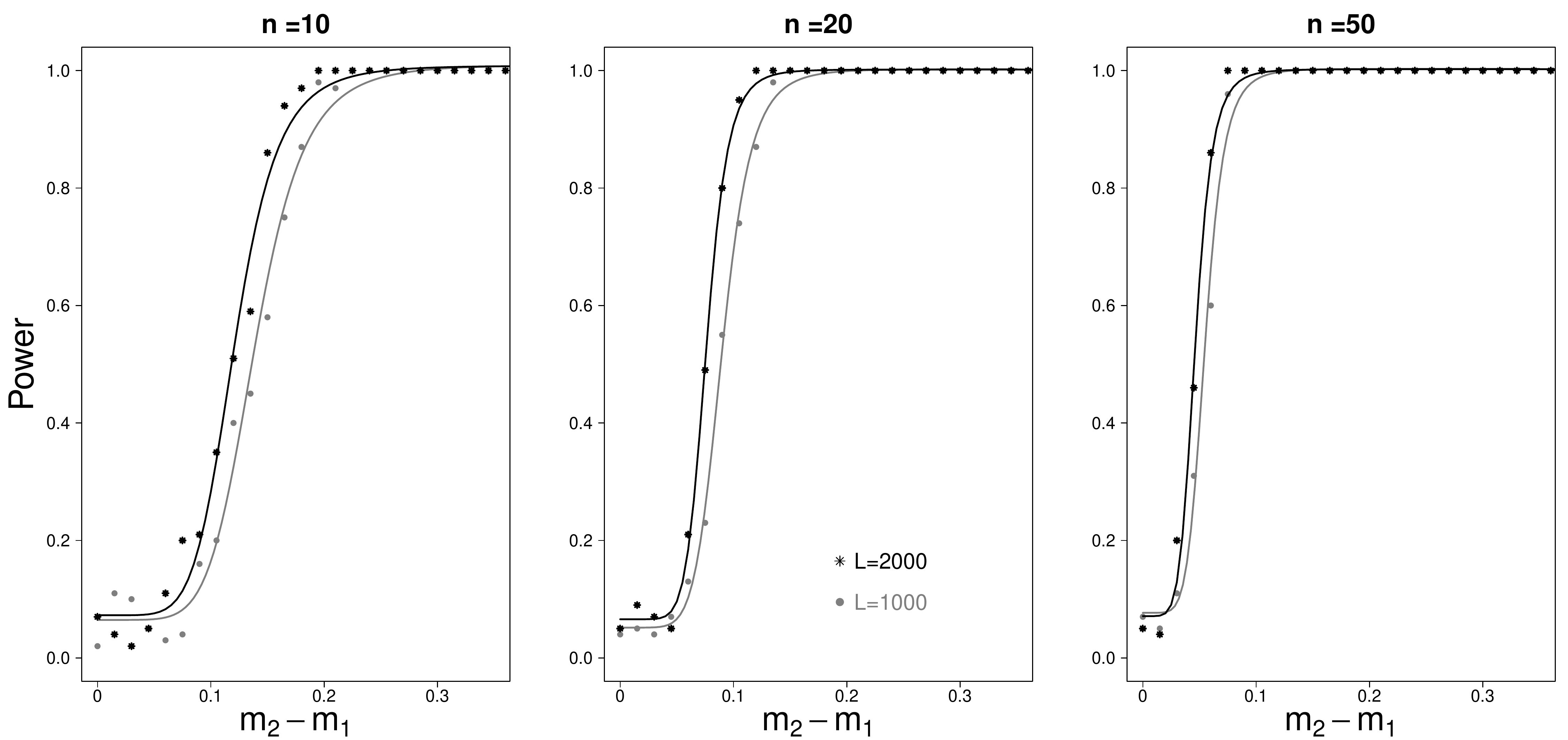}
\end{center} \vspace{-2mm}
\caption{Power curves of $utest$ for two dimension $L=1000$ (gray) and $L=2000$ (black) for 100 replications of each scenario of $n \in \{10,20,50\}$ with $\alpha=0.05$.
}
\label{fig:pw1}
\end{figure}

The empirical results obtained in this study reported in Figure \ref{fig:pw1} corroborate the theoretical properties. As the $L$ increases, the rejection ratio also increases  and as the groups become more separated, the power increases. When there is no separation, $m_2-m_1=0$, the rejection ratio is close to the significance level $\alpha$ suggesting control of Type I error. Similar results are found for cases where all groups have more than one element (see Figure S1 in the Supplementary Material).

\pagebreak
\subsection{Simulations for the homogeneity test in $uclust3$}\label{subsec:ht}

To evaluate the statistical properties of the homogeneity test $uclust3$ considering the max distribution \eqref{eq:testmaxBn} with the Gumbel correction when appropriate, we simulate data  with the same characteristics as the data in Section \ref{subsec:ut}. For each sample size $n$ in $\{10,20,50\}$, group $G_1$ has size one and group $G_2$ was set to have size $n_2=2$ and $n_2=n/2$, and consequently the  third group's size was defined as $n_3=n-1-n_2$. Table \ref{tab:pw1} shows the proportion of rejection of the null hypothesis for significance level $\alpha=0.05$ considering two scenarios of $(m_2,m_3)$ and the dimension $L$ taking values in $\{1000,2000\}$.

\begin{table}[htb!]
\centering
\caption{Empirical power of the homogeneity test $uclust3$  with a group of size one}
\label{tab:pw1}
\begin{tabular}{c|cc|cc}
\toprule
& \multirow{2}{*}{$(m_2,\,\,m_3)$} & \multirow{2}{*}{($n_2$)} & \multicolumn{2}{c}{Dimension $L$} \\
\cline{4-5}
$n$ & & & 1000 & 2000 \\
\cline{4-5} 
\midrule
\multirow{4}{*}{10} & \multirow{2}{*}{(0.25, 0.5)} & 2 & 0.27 & 0.36 \\
\cline{3-5}
& & 5 & 0.69 & 0.89 \\
\cline{2-5}
& \multirow{2}{*}{(0.5,  1)} & 2 & 0.22 & 0.25 \\
\cline{3-5}
& & 5 & 0.98 & 1 \\
\midrule
\multirow{4}{*}{20} & \multirow{2}{*}{(0.25, 0.5)} & 2 & 0.93 & 1 \\
\cline{3-5}
& & 10 & 1 & 1 \\
\cline{2-5}
& \multirow{2}{*}{(0.5,  1)} & 2 & 0.9 & 0.89 \\
\cline{3-5}
& & 10 & 0.92 & 1 \\
\midrule
\multirow{4}{*}{50} & \multirow{2}{*}{(0.25, 0.5)} & 2 & 0.68 & 0.68 \\
\cline{3-5}
& & 25 & 1 & 0.99 \\
\cline{2-5}
& \multirow{2}{*}{(0.5,  1)} & 2 & 0.99 & 0.96 \\
\cline{3-5}
& & 25 & 1 & 1 \\
\bottomrule
\end{tabular}
\end{table}

We can observe that even in an extreme group configuration, where the group  $G_1$ has size one and the group $G_2$ has size two,  the method presents consistent empirical power to reject the null hypothesis. The power increases as $L$ and/or $n$ and/or the difference between $m_2$ and $m_3$ increases, emphasizing the inherent properties of the method. 

Supplementary Table S1 presents estimates of type I error rates for $uclust3$. The significance level considered in this simulations was $\alpha=0.05$ and we can observe that the method presents an adequate control of the Type I Error for cases where $L>>n$ (typically HDLSS scenario). Supplementary Table S2 presents power of the  $uclust3$ for group configurations of sizes greater than 1. For small sample size $n$ the test had more difficulty in finding the correct clusters. However, for larger $n$ the method showed an excellent performance.

\subsection{Simulations for finding correct clusters}\label{subsec:cc}

In order to evaluate the accuracy of our clustering method, we present simulation studies comparing $uclust3$ with $kmeans$ clustering, one of the most popular clustering algorithms. We refer the reader to the vastly cited work of \cite{jain2010}  for a general discussion about $kmeans$. The data were simulated under the same distribution scheme of Section \ref{subsec:ht}, with $Re=100$ replications and the methods were compared in terms of mean Adjusted Rand Index (ARI) which measures the agreement of clustering results with simulation scenarios, adjusting for randomness \cite{hubert1985}. An ARI of one indicates perfect matching. No inference is used in this analysis.  This is an appropriate comparison as both methods are set to find exactly three groups. Table \ref{tab:arikmeans} reports the results for three sample sizes $n\in\{10,20,50\}$, two dimension $L\in\{1000,2000\}$ and three groups of sizes $n_1$, $n_2$ and $n_3=n-n_1-n_2$. The data vectors in group $G_1$ have zero mean and the data vectors in $G_2$ and $G_3$ have mean $m_2$ and $m_3$, respectively. Note that the clustering method $uclust3$, based on the maximization of $B_n$ is comparable to $kmeans$ to find the correct clusters, considering this data configuration.
However for larger sample sizes, as the clusters become better defined, with greater separation between the means, $uclust3$ outperforms $kmeans$. Table S3 shows that for the case where $G_1$ has size one, $kmeans$ tends to perform slightly better for smaller sample sizes.

\begin{table}[htb!]
\centering
\caption{Comparison of mean ARI and standard deviation (Sd)  of the accuracy in clustering of $kmeans$ and $uclust3$ methods.}\label{tab:arikmeans}
\label{ari}
\begin{tabular}{c|cc|ccccc}
\toprule
 & \multirow{3}{*}{ $(m_2,\,\,m_3)$ } & \multirow{3}{*}{ ($n_1, \,\,n_2$) } & \multirow{3}{*}{Method} & \multicolumn{4}{c}{Dimension $L$} \\
\cline{5-8}
$n$ & & & & \multicolumn{2}{c}{ 1000 } & \multicolumn{2}{c}{ 2000 } \\
\cline{5-8}
 & & & & Mean & Sd & Mean & Sd \\ 
\midrule
\multirow{8}{*}{ 10 } & \multirow{4}{*}{ (0.25, 0.5) } & \multirow{2}{*}{ (2, 5) } & $kmeans$ &  0.59  &  0.05  &  0.73  &  0.06  \\
& & & $uclust3$ &  0.58  &  0.03  &  0.63  &  0.02  \\
\cline{3-8}
& & \multirow{2}{*}{ (3, 3) } & $kmeans$  &  0.56  &  0.05  &  0.74  &  0.08  \\
& & & $uclust3$ &  0.52  &  0.05  &  0.6  &  0.05  \\
\cline{2-8}
& \multirow{4}{*}{ (0.5,  1) } & \multirow{2}{*}{ (2, 5) } & $kmeans$ &  0.91  &  0.04  &  0.94  &  0.03  \\
& & & uclust3 &  0.74  &  0.01  &  0.74  &  0  \\
\cline{3-8}
& & \multirow{2}{*}{ (3, 3) } & $kmeans$ &  0.9  &  0.05  &  0.87  &  0.07  \\
& & & $uclust3$ &  0.92  &  0.03  &  0.96  &  0.02  \\
\midrule
\multirow{8}{*}{ 20 } & \multirow{4}{*}{ (0.25, 0.5) } & \multirow{2}{*}{ (2, 10) } & $kmeans$ &  0.73  &  0.02  &  0.77  &  0.03  \\
& & & $uclust3$ &  0.7  &  0.02  &  0.74  &  0.02  \\
\cline{3-8}
& & \multirow{2}{*}{ (6, 6) } & $kmeans$  &  0.74  &  0.05  &  0.94  &  0.03  \\
& & & $uclust3$ &  0.68  &  0.04  &  0.91  &  0.02  \\
\cline{2-8}
& \multirow{4}{*}{ (0.5,  1) } & \multirow{2}{*}{ (2, 10) } & $kmeans$ &  0.96  &  0.01  &  0.94  &  0.02  \\
& & & $uclust3$ &  1  &  0  &  1  &  0  \\
\cline{3-8}
& & \multirow{2}{*}{ (6, 6) } & $kmeans$ &  0.81  &  0.07  &  0.84  &  0.07  \\
& & & $uclust3$ &  1  &  0  &  1  &  0  \\
\midrule 
\multirow{8}{*}{ 50 } & \multirow{4}{*}{ (0.25, 0.5)} & \multirow{2}{*}{ (2, 25) } & $kmeans$ &  0.76  &  0.01  &  0.79  &  0.01  \\
& & & $uclust3$ &  0.73  &  0  &  0.74  &  0.01  \\
\cline{3-8}
& & \multirow{2}{*}{ (16, 16) } & $kmeans$  &  0.93  &  0.02  &  0.89  &  0.05  \\
& & & $uclust3$ &  0.94  &  0  &  1  &  0  \\
\cline{2-8}
& \multirow{4}{*}{(0.5 ,  1) } & \multirow{2}{*}{ (2 , 25) } & $kmeans$ &  0.95  &  0.01  &  0.95  &  0.01  \\
& & & $uclust3$ &  1  &  0  &  1  &  0  \\
\cline{3-8}
& & \multirow{2}{*}{ (16, 16) } & $kmeans$ &  0.8  &  0.07  &  0.81  &  0.07  \\
& & & $uclust3$ &  1  &  0  &  1  &  0  \\
\bottomrule
\end{tabular}
\end{table}

\pagebreak
\subsection{Finding correct clusters and comparing $uclust3$ and $uhclust$ in a presence of an outlier}\label{subsec:compuhsig}

A simulation study similar to Section \ref{subsec:ut} was performed to compare our $uclust3$ with the hierarchical methods $uhclust$ from \cite{valk20} and $sigclust$ from \cite{kimes17,sigclust2} in terms of the ability to correctly find  statistically significant groups. The group $G_1$ has only one element, the size of $G_2$ is $n_2=\lfloor n/3\rfloor$. For all three methods the same level of significance $\alpha=0.05$ was considered. The $sigclust$ method was not able to find the correct groups in any scenario, with a proportion of correct answers equal to zero and for this reason it was excluded from the analysis.  Figures \ref{fig:proptrueclust11} and \ref{fig:proptrueclust12} report curves of proportion times that the algorithms found significant separation and correct groups considering different values of $m_2-m_1$ varying on the $x$ axis, with sample size $n$ taking values in $\{10,20,50\}$ and dimension $L=1000$ and $L=2000$ The results are based on 50 repetitions.

\begin{figure}[h!]
\begin{center}
\includegraphics[scale=0.25]{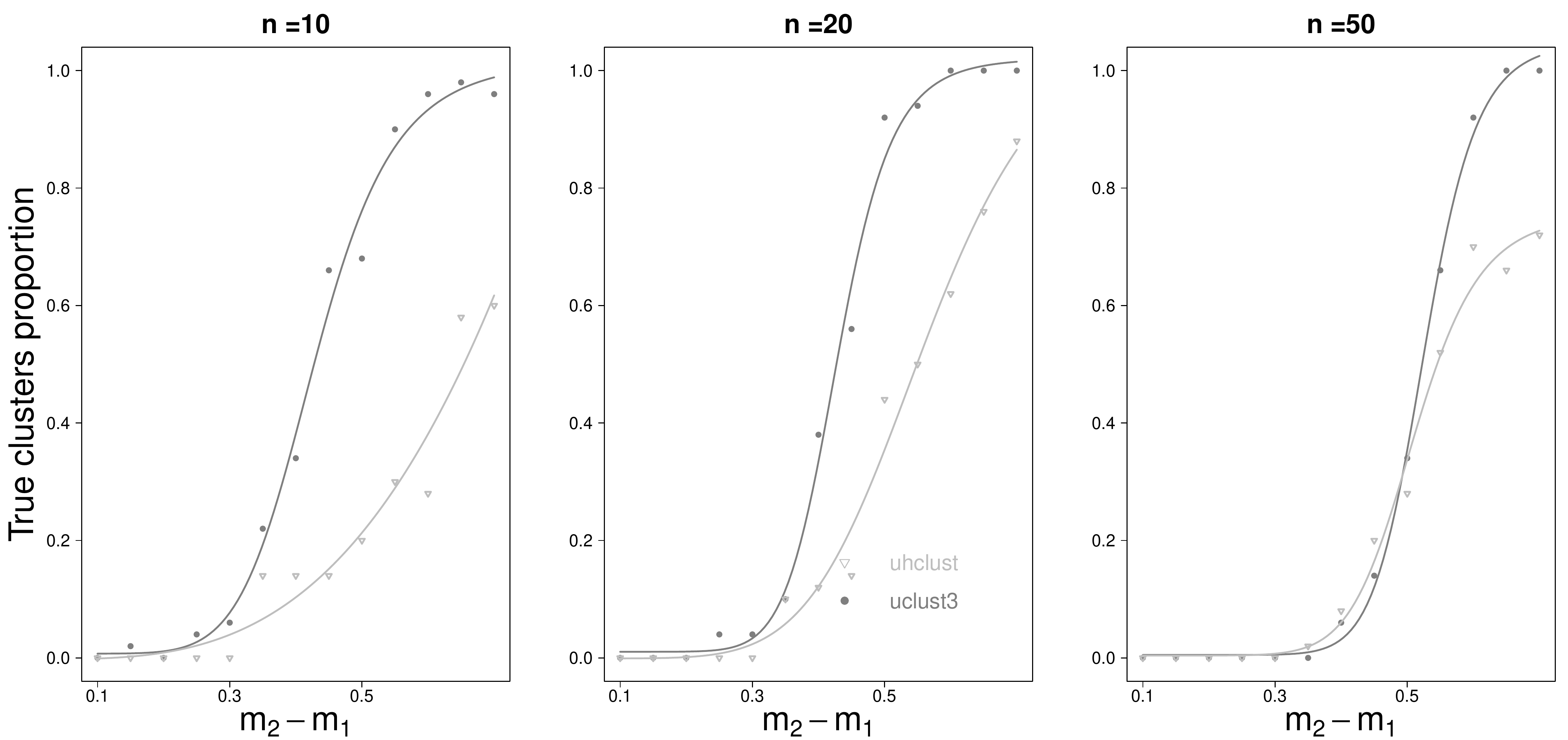}
\end{center} \vspace{-2mm}
\caption{True cluster proportion curves of $uclust3$ (dark gray) and $uhclust$ (light gray) for dimension $L=1000$ with 50 replications of each scenario of $n$ with $\alpha=0.05$ and one outlier.
}
\label{fig:proptrueclust11}
\end{figure}

\begin{figure}[h!]
\begin{center}
\includegraphics[scale=0.25]{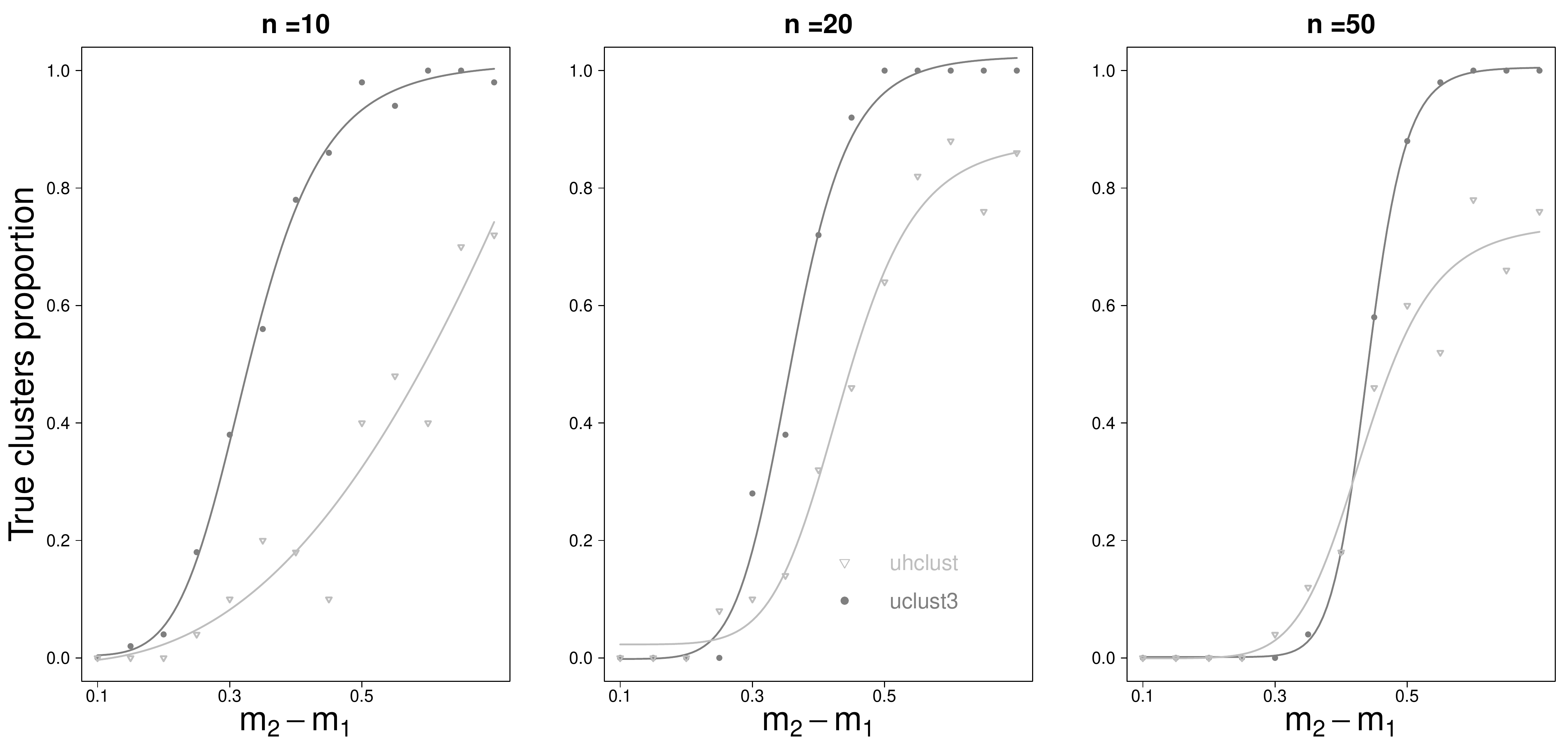}
\end{center} \vspace{-2mm}
\caption{True cluster proportion curves of $uclust3$ (dark gray) and $uhclust$ (light gray)  for dimension $L=2000$ with 50 replications of each scenario of $n$ with $\alpha=0.05$ and one outlier.
}
\label{fig:proptrueclust12}
\end{figure}

 The $uclust3$ method (dark grey)  outperforms $uhclust$ method (light gray) in all scenarios presenting greater ability to find the correct groups for less separation. However, for $n = 50$ these method are more competitive  although the method proposed here $uclust3$ still stands out for larger separations. The conclusions do not change with the variation of dimension $L$. In Section S5 on the supplementary materials we present results of a simulation study for the cases where there are no outlier.  Supplementary Figures S2 and S3 shows the true cluster proportion curves of $uclust3$ and $uhclust$  for dimension $L=1000$ and $L=2000$. We note that the $uclust3$ method outperforms $uhclust$ in all scenarios.


\section{Applications}\label{sec::application}

\subsection{Peripheral blood mononuclear cells }
In order to illustrate of the applicability of the \emph{utest} we consider a one-way MANOVA (multivariate analysis of variance) testing problem for high-dimensional data. This issue was addressed in \cite{zhang2017} by exploring peripheral blood mononuclear cell (PBMC) data, consisting of 42 normal, 26 ulcerative colitis (UC) and 59 Crohn’s disease (CD) tissue samples ($n=127$), each having $L=22,283$ gene expression level measurements. This dataset has been studied by \cite{burczynski2006} and is available at \emph{http://www.ncbi.nlm.nih.gov/gds} with accession ID GDS1615.  The classical hypothesis test where
the interest is to test whether the 3 mean vectors are equal, can be described as follows: Let $\mathbf{X}^{(g)}_{1}, \ldots, \mathbf{X}^{(g)}_{n_{g}}$ be a sample of i.i.d. vectors from the $L$-variate distribution ${F}_g$, with $ \E(\mathbf{X}^{(g)}_{1})=\bm{\mu}_g$ and  $\cov(\mathbf{X}^{(g)}_{1}) = \bm{\Sigma},\quad $  for $ g=1, \ldots, 3$ and  $ n=n_1+n_2+n_3$.
%
%
Then, the null hypothesis is
$H_0:\bm{\mu}_1=\bm{\mu}_2=\bm{\mu}_3$. In our context, however, the normality and variance homogeneity requirements are not necessary, and the null hypothesis becomes the more general
\[H_0:{F}_1={F}_2={F}_3.\]
We apply the $utest$ for testing the equality of mean expression levels of the normal, UC and CD groups of the PBMC data.  The value of standardized $B_n$ statistic is 13.20997 (p-value$<<$0.001) with which we reject the null hypothesis of equality of mean expression levels.

\label{subsec::application_pbmc}

\subsection{Image recognition}\label{subsec::application_image}

We consider a simple example of image recognition to illustrate the applicability of our methodology. The data consists of  images from three public figures (Tony Blair, Colin Powell and George W. Bush) which were selected from the Labeled Faces Wild (LFW) dataset (\cite{lfw}). The data were run through OpenFace's convolutional neural network (\cite{openface}), a procedure that outputs a 128-dimensional representation of the faces which preserves Euclidean distances. In case the reader wants to know more about how the OpenFace works, we recommend reading their website \cite{openface}. In this illustrative application, we randomly select 10 images from each  public figure in the above cited dataset, run $uhclust$, $sigclust$ and $uclust3$ with significance level $\alpha=0.05$. Figure \ref{fig:uhc} presents the hierarchical clustering dendrogram annotated with p-values for all tests performed in the $uhclust$ method. We found 4 homogeneous groups, with a significant division in the Bush image group and an ARI=0.8585. Figure \ref{fig:sig} presents the dendrogram with corresponding $sigclust$ analysis of the same data which produces six significant clusters, segregating Bush and Powell's images from the reminder and finding one outlier in Blair's group. The ARI for this case was 0.7788. Applying the $uhclust3$ method we found exactly 3 homogeneous groups, each corresponding to one of the public figures with ARI=1.  

In the Section S6 in the supplementary materials we consider the same dataset and public figures to carry out an analysis with three groups in which one has size one. Figures S4 and S5 in the supplementary materials present the clustering dendrogram annotated with results of  all tests performed in the $uhclust$ and $sigclust$ methods. None of these methods were able to identify the outlier and both methods achieved ARI of 0.8135593. However, when we applied the $uclust3$ method we found the correct groups with ARI of 1, supporting the best results $uclust3$ in the simulation study.

\begin{figure}[h!]
\begin{center}
\includegraphics[scale=0.45]{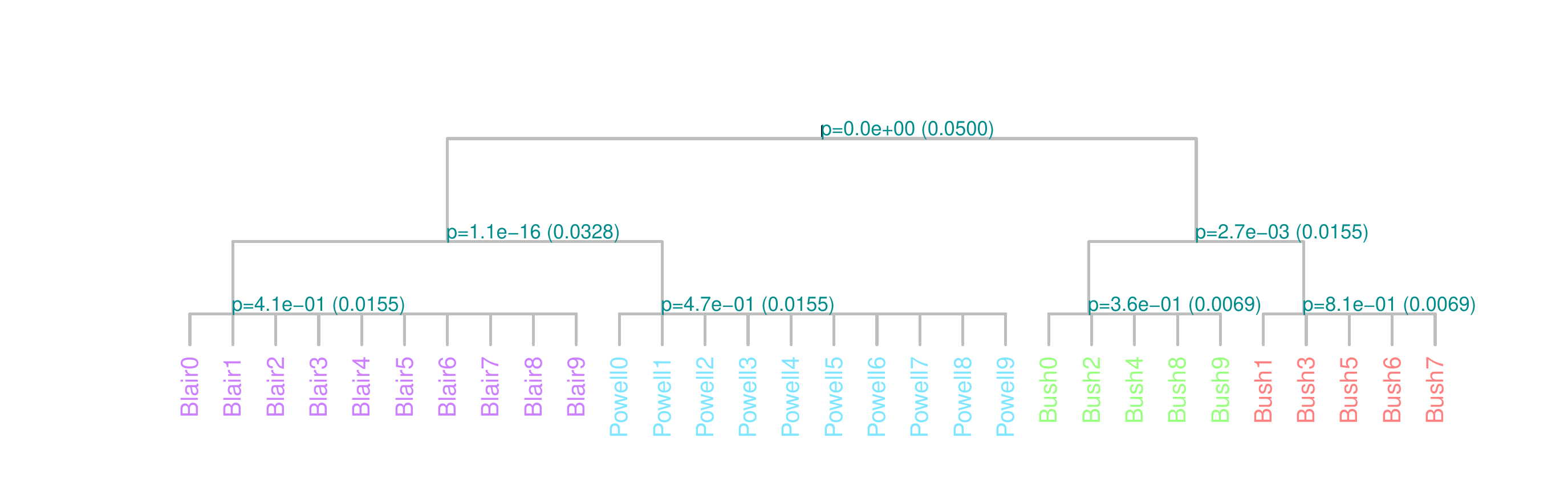}
\end{center} \vspace{-2mm}
\caption{Annotated dendrogram of significance analysis for hierarchical clustering $uhclust$ for 30 pictures of 3 public figures. P-values and corrected significance levels $\alpha^*$ are shown for each test performed at the corresponding node.
}
\label{fig:uhc}
\end{figure}

\begin{figure}[h!]
\begin{center}
\includegraphics[scale=0.42]{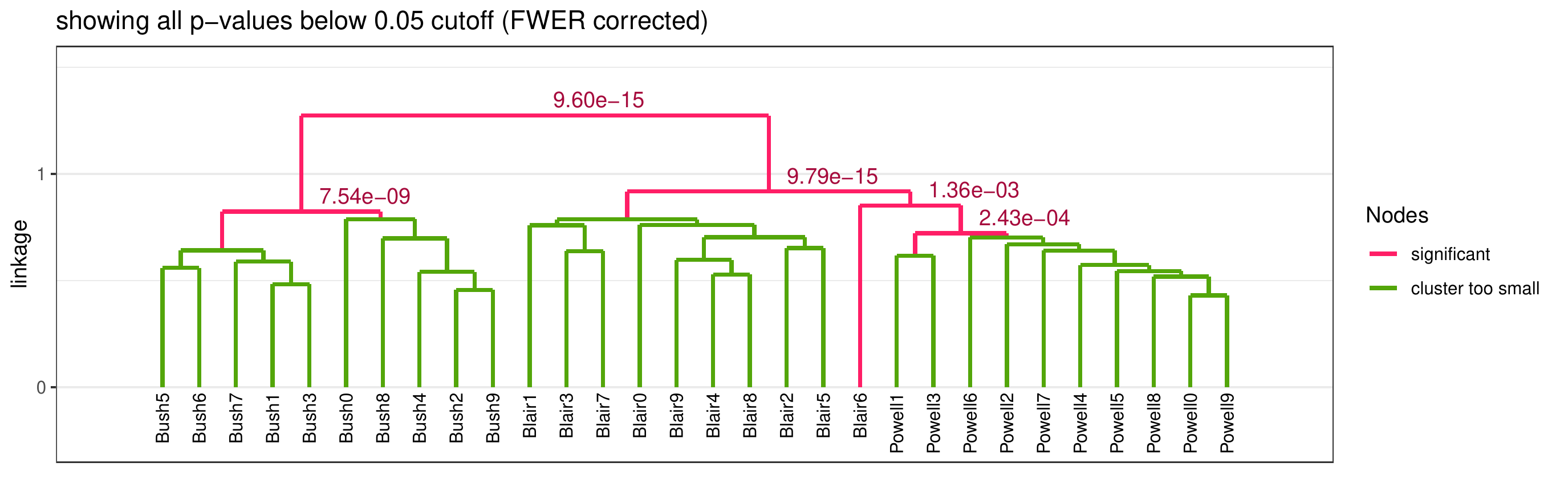}
\end{center} \vspace{-2mm}
\caption{Annotated dendrogram of significance analysis for hierarchical clustering $sigclust$ for 30 pictures of 3 public figures. P-values and corrected significance levels $\alpha^*$ are shown for each test performed at the corresponding node.
}
\label{fig:sig}
\end{figure}

\pagebreak
\section{Discussion}\label{sec:conclusions}

We have developed a clustering method that separates a dataset specifically into \emph{three} groups allowing the assessment of significance of this partition.  Our methodology is based on the U-statistics clustering framework proposed in \cite{pinheiro09} and is an extension of the approach of \cite{cybis18,valk20}. 
Considering the $B_n$ statistic of \cite{pinheiro09} that aims
to test homogeneity of three predefined groups we propose an extension of the $B_n$ statistic to allow for an outlier, namely one of the groups has only one element ($n_1=1$). Additionally we verified statistical properties that ensure the compatibility of this new definition with the overall framework. 
We then considered group homogeneity testing with this newly defined statistic, and explored empirical properties such as Type I error control and power, 
showing adequate preformance. Afterwards, we extended this framework to address the issue of partitioning a dataset into three optimal statistically significant clusters, proposing a new clustering criteria that defines the $uclust3$ method. This differs from previous methods for instead of find and testing a two group separation, $uclust3$ finds the best significant partitions in three clusters. This can pave the way for inference in $K$ groups.

This U-statistics based methodology can be applied to a wide range of problems, since they make very few assumptions about the distribution of the data. Although in the simulation study and in the application we have used Euclidean distance, this is not a necessary requirement  for theory development. Additionally, even if the data come from a non-normal multivariate distribution, the required asymptotic normality is guaranteed as long as the distances have finite variance and the sum of all distance covariances do not grow too fast ($O(L)$ see Theorem 2).  The clustering procedures $uclust3$ proposed here require large $L$ since $B_n$ for $n_1=1$ is only asymptotically normal in the dimension $L$. As verified in previously work of \cite{valk20}, for the settings in the simulation studies, in practice our tests achieve good Type I error control having difficulties only when $L$ is smaller than $10n$. This is, by excellence, the HDLSS setting.

An important step for developing the homogeneity test is to establish the number of possible configurations of $n$ elements separated in three groups. A system of recursive equations was developed to solve this combinatorial problem and the idea may be used to solve an equivalent problem involving $K>3$ groups.

The significance clustering method $uclust3$ proposed here returns the partition that better separates the data into three statistically significant groups in terms of the $B_n$ statistic. Thus we can compare it with $kmeans$, which is one of the most popular clustering method,  regarding the ability of correctly find three groups. A simulation study suggests that $uclust3$ is competitive with $kmeans$ when we have a size one group and outperforms $kmeans$ in the context in which groups having an underlying cluster structure with more than 2 elements each and large sample sizes.

Since our methodology is a natural extension of the $uclust$ method proposed by \cite{valk20} it inherits many  helpful properties such as the ability to avoid the hazards of directly estimating the covariance matrix, by obtaining $\operatorname{Var}(B_n)$ through  resampling. However, they have different purposes, while $uclust$ aims to find the best significant partition in two groups, $uclust3$ aims to find the best significant separation in three groups, so they are not directly comparable. To support the usefulness of the $uclust3$, we carried out a simulation study to compare this method with the hierarchical version of $uclust$ ($uhclsut$) and with another hierarchical approach ($sigclust$), which both are able to find a significant partition into three groups, when this partition exists. We simulated normal data with a three group structure, separating these groups in terms of the means and use the  proportion of correct configurations found to compare the methods. In the situations considered, $sigclust$ had serious difficulties in finding the proper arrangement, while $uclust3$ performed better than $uhclust$ in all scenarios. Additionally, in the applications we have shown the applicability of this methodology, first with a one-way MANOVA testing problem without the requirement of normality of data and variance homogeneity, and then with an application to image recognition data where we select three public figures and observe that the $uclust3$ method was the only one able to correctly find the three groups of figures.

Finally the conclusion is that our $uclust3$ method is appropriate to separate a high dimensional low sample size datasets into \emph{three} groups, being more powerful than some other methods in the specific situation in which a structure of three groups is present in the data.

\section*{Supplementary material}

\begin{description}
\item[Supplementary material:] Derivations, supplementary tables and figures (pdf)
\item[Code:] R-functions containing all methods developed in this article (will be available in the uclust package at CRAN).
\item[Data:] Dataset used in the application and corresponding script (zip).
\end{description}

\section*{ Acknowledgements}
\small
We would like to thanks Pedro Fusieger for the fruitful discussions about the number of possible assignments of all $n$ elements in three subgroups.

\section*{ Funding}
\small
Debora Zava Bello research was supported by Coordenação de Aperfeiçoamento de Pessoal de Nível Superior (CAPES).

\bibliographystyle{unsrtnat}
\bibliography{references}

\begin{thebibliography}{28}
\providecommand{\natexlab}[1]{#1}
\providecommand{\url}[1]{\texttt{#1}}
\expandafter\ifx\csname urlstyle\endcsname\relax
  \providecommand{\doi}[1]{doi: #1}\else
  \providecommand{\doi}{doi: \begingroup \urlstyle{rm}\Url}\fi

\bibitem[Euan et~al.(2019)Euan, Sun, Ombao, et~al.]{euan2019}
Carolina Euan, Ying Sun, Hernando Ombao, et~al.
\newblock Coherence-based time series clustering for statistical inference and
  visualization of brain connectivity.
\newblock \emph{Annals of Applied Statistics}, 13\penalty0 (2):\penalty0
  990--1015, 2019.

\bibitem[Rosenberg et~al.(2002)Rosenberg, Pritchard, Weber, Cann, Kidd,
  Zhivotovsky, and Feldman]{rosenberg2002}
Noah~A. Rosenberg, Jonathan~K. Pritchard, James~L. Weber, Howard~M. Cann,
  Kenneth~K. Kidd, Lev~A. Zhivotovsky, and Marcus~W. Feldman.
\newblock Genetic structure of human populations.
\newblock \emph{Science}, 298\penalty0 (5602):\penalty0 2381--2385, 2002.

\bibitem[Chen et~al.(2015)Chen, Chi, Ranola, and Lange]{chen2015}
Gary~K. Chen, Eric~C. Chi, John Michael~O. Ranola, and Kenneth Lange.
\newblock Convex clustering: an attractive alternative to herarchical
  clustering.
\newblock \emph{PLoS Computational Biology}, 11\penalty0 (5):\penalty0
  e1004228, 2015.

\bibitem[Motlagh et~al.(2019)Motlagh, Berry, and O'Neil]{motlagh2019}
Omid Motlagh, Adam Berry, and Lachlan O'Neil.
\newblock Clustering of residential electricity customers using load time
  series.
\newblock \emph{Applied energy}, 237:\penalty0 11--24, 2019.

\bibitem[Hennig(2015)]{hennig2015}
Christian Hennig.
\newblock What are the true clusters?
\newblock \emph{Pattern Recognition Letters}, 64:\penalty0 53--62, 2015.

\bibitem[Von~Luxburg et~al.(2012)Von~Luxburg, Williamson, and Guyon]{von2012}
Ulrike Von~Luxburg, Robert~C Williamson, and Isabelle Guyon.
\newblock Clustering: Science or art?
\newblock In \emph{Proceedings of ICML workshop on unsupervised and transfer
  learning}, pages 65--79. JMLR Workshop and Conference Proceedings, 2012.

\bibitem[Adolfsson et~al.(2019)Adolfsson, Ackerman, and
  Brownstein]{adolfsson2019}
Andreas Adolfsson, Margareta Ackerman, and Naomi~C Brownstein.
\newblock To cluster, or not to cluster: An analysis of clusterability methods.
\newblock \emph{Pattern Recognition}, 88:\penalty0 13--26, 2019.

\bibitem[McLachlan and Peel(2004)]{mclachlan2004}
Geoffrey McLachlan and David Peel.
\newblock \emph{Finite mixture models}.
\newblock John Wiley \& Sons, 2004.

\bibitem[Demidenko(2018)]{demidenko2018}
Eugene Demidenko.
\newblock The next-generation k-means algorithm.
\newblock \emph{Statistical Analysis and Data Mining: The ASA Data Science
  Journal}, 11\penalty0 (4):\penalty0 153--166, 2018.

\bibitem[McShane et~al.(2002)McShane, Radmacher, Freidlin, Yu, Li, and
  Simon]{mcshane2002}
Lisa~M McShane, Michael~D Radmacher, Boris Freidlin, Ren Yu, Ming-Chung Li, and
  Richard Simon.
\newblock Methods for assessing reproducibility of clustering patterns observed
  in analyses of microarray data.
\newblock \emph{Bioinformatics}, 18\penalty0 (11):\penalty0 1462--1469, 2002.

\bibitem[Helgeson et~al.(2020)Helgeson, Vock, and Bair]{helgeson2020}
Erika~S Helgeson, David~M Vock, and Eric Bair.
\newblock Nonparametric cluster significance testing with reference to a
  unimodal null distribution.
\newblock \emph{Biometrics}, 2020.

\bibitem[Shimodaira et~al.(2004)]{shimodaira2004}
Hidetoshi Shimodaira et~al.
\newblock Approximately unbiased tests of regions using multistep-multiscale
  bootstrap resampling.
\newblock \emph{The Annals of Statistics}, 32\penalty0 (6):\penalty0
  2616--2641, 2004.

\bibitem[Suzuki and Shimodaira(2006)]{suzuki2006}
Ryota Suzuki and Hidetoshi Shimodaira.
\newblock Pvclust: an {R} package for assessing the uncertainty in hierarchical
  clustering.
\newblock \emph{Bioinformatics}, 22\penalty0 (12):\penalty0 1540--1542, 2006.

\bibitem[Liu et~al.(2008)Liu, Hayes, Nobel, and Marron]{Liu2008}
Yufeng Liu, David~Neil Hayes, Andrew Nobel, and JS~Marron.
\newblock Statistical significance of clustering for high-dimension,
  low--sample size data.
\newblock \emph{Journal of the American Statistical Association}, 103\penalty0
  (483):\penalty0 1281--1293, 2008.

\bibitem[Kimes et~al.(2017)Kimes, Liu, Neil~Hayes, and Marron]{kimes17}
Patrick~K Kimes, Yufeng Liu, David Neil~Hayes, and James~Stephen Marron.
\newblock Statistical significance for hierarchical clustering.
\newblock \emph{Biometrics}, 73\penalty0 (3):\penalty0 811--821, 2017.

\bibitem[Cybis et~al.(2018)Cybis, Valk, and Lopes]{cybis18}
Gabriela~B. Cybis, Marcio Valk, and Sílvia R.~C. Lopes.
\newblock Clustering and classification problems in genetics through
  u-statistics.
\newblock \emph{Journal of Statistical Computation and Simulation}, 2018.

\bibitem[Valk and Cybis(2020)]{valk20}
Marcio Valk and Gabriela~Bettella Cybis.
\newblock U-statistical inference for hierarchical clustering.
\newblock \emph{Journal of Computational and Graphical Statistics}, 2020.

\bibitem[Pinheiro et~al.(2009)Pinheiro, Sen, and Pinheiro]{pinheiro09}
Aluísio Pinheiro, Pranab~Kumar Sen, and Hildete~Prisco Pinheiro.
\newblock Decomposability of high-dimensional diversity measures:
  Quasi-u-statistics, martingales and nonstandard asymptotics.
\newblock \emph{Journal of Multivariate Analysis}, 2009.

\bibitem[Sen(2006)]{Sen2006}
Pranab~Kumar Sen.
\newblock Robust statistical inference for high-dimensional data models with
  application to genomics.
\newblock \emph{Austrian journal of statistics}, 35\penalty0 (2\&3):\penalty0
  197--214, 2006.

\bibitem[Hoeffding(1948)]{Hoeffding1948}
Wassily Hoeffding.
\newblock A class of statistics with asymptotically normal distribution.
\newblock \emph{The Annals of Mathematical Statistics}, pages 293--325, 1948.

\bibitem[Valk and Pinheiro(2012)]{valk12}
Marcio Valk and Aluísio Pinheiro.
\newblock Time-series clustering via quasi u-statistics.
\newblock \emph{Journal of Time Series Analysis}, 2012.

\bibitem[Jain(2010)]{jain2010}
Anil~K Jain.
\newblock Data clustering: 50 years beyond k-means.
\newblock \emph{Pattern recognition letters}, 31\penalty0 (8):\penalty0
  651--666, 2010.

\bibitem[Hubert and Arabie(1985)]{hubert1985}
Lawrence Hubert and Phipps Arabie.
\newblock Comparing partitions.
\newblock \emph{Journal of classification}, 2\penalty0 (1):\penalty0 193--218,
  1985.

\bibitem[Kimes(2019)]{sigclust2}
Patrick Kimes.
\newblock pkimes/sigclust2 documentation, 2019.
\newblock URL \url{https://rdrr.io/github/pkimes/sigclust2/man/}.
\newblock (Accessed on 02/02/2021).

\bibitem[Zhang et~al.(2017)Zhang, Guo, and Zhou]{zhang2017}
Jin-Ting Zhang, Jia Guo, and Bu~Zhou.
\newblock Linear hypothesis testing in high-dimensional one-way manova.
\newblock \emph{Journal of Multivariate Analysis}, 155:\penalty0 200--216,
  2017.

\bibitem[Burczynski et~al.(2006)Burczynski, Peterson, Twine, Zuberek, Brodeur,
  Casciotti, Maganti, Reddy, Strahs, Immermann, et~al.]{burczynski2006}
Michael~E Burczynski, Ron~L Peterson, Natalie~C Twine, Krystyna~A Zuberek,
  Brendan~J Brodeur, Lori Casciotti, Vasu Maganti, Padma~S Reddy, Andrew
  Strahs, Fred Immermann, et~al.
\newblock Molecular classification of crohn's disease and ulcerative colitis
  patients using transcriptional profiles in peripheral blood mononuclear
  cells.
\newblock \emph{The journal of molecular diagnostics}, 8\penalty0 (1):\penalty0
  51--61, 2006.

\bibitem[Huang et~al.(2007)Huang, Ramesh, Berg, and Learned-Miller]{lfw}
Gary~B. Huang, Manu Ramesh, Tamara Berg, and Erik Learned-Miller.
\newblock Labeled faces in the wild: A database for studying face recognition
  in unconstrained environments.
\newblock Technical Report 07-49, University of Massachusetts, Amherst, October
  2007.

\bibitem[Amos et~al.(2016)Amos, Ludwiczuk, and Satyanarayanan]{openface}
Brandon Amos, Bartosz Ludwiczuk, and Mahadev Satyanarayanan.
\newblock {O}pen{F}ace: A general-purpose face recognition library with mobile
  applications.
\newblock Technical report, CMU-CS-16-118, CMU School of Computer Science,
  2016.

\end{thebibliography}



\section*{Supplementary material}

# Supplementary material: Clustering inference in multiple groups

Debora Zava Bello, Marcio Valk and Gabriela Cybis

Jun 2021

## S1 The extended $B_n$ for three groups

In this work we propose an extension of the statistic $B_n$ for three groups allowing for a size one group. This extension, as shown in the Section 2.2 of the main manuscript, was defined as

$$B_n = \begin{cases} \frac{2n_2}{n(n-1)}\left(U_{1,n_2}^{(1,2)} - U_{n_2}^{(2)}\right) + \frac{2n_3}{n(n-1)}\left(U_{1,n_3}^{(1,3)} - U_{n_3}^{(3)}\right) \\ +\frac{n_2 n_3}{n(n-1)}\left(2U_{n_2,n_3}^{(2,3)} - U_{n_2}^{(2)} - U_{n_3}^{(3)}\right), & \text{if } n_1 = 1, \text{ and } n_2, n_3 > 1 \\ \displaystyle\sum_{1 \leq i < j \leq 3} \frac{n_i n_j}{n(n-1)}\left(2U_{n_i,n_j}^{(i,j)} - U_{n_i}^{(i)} - U_{n_j}^{(j)}\right), & \text{if } n_1, n_2, n_3 > 1. \end{cases} \quad (S.1)$$

where $U_{n_g,n_{g'}}^{(g,g')}$ and $U_{n_g}^{(g)}$ are defined, respectively, in equations (3) and (4) in the manuscript. As properties of $B_n$ are well described for cases where groups have more than one element we focus on the special case in which one of the groups has size one. Without loss of generality assume that $n_1 = 1$ and $n_2, n_3 > 1$. Thus $B_n$ becomes

$$\begin{aligned} B_n &= \frac{2n_2 U_{1,n_2}^{(1,2)}}{n(n-1)} - \frac{2n_2 U_{n_2}^{(2)}}{n(n-1)} + \frac{2n_3 U_{1,n_3}^{(1,3)}}{n(n-1)} - \frac{2n_3 U_{n_3}^{(3)}}{n(n-1)} + \frac{2n_2 n_3 U_{n_2,n_3}^{(2,3)}}{n(n-1)} \\ &\quad -\frac{n_2 n_3 U_{n_2}^{(2)}}{n(n-1)} - \frac{n_2 n_3 U_{n_3}^{(3)}}{n(n-1)} \\ &= \frac{2n_2 U_{1,n_2}^{(1,2)}}{n(n-1)} + \frac{2n_3 U_{1,n_3}^{(1,3)}}{n(n-1)} + \frac{2n_2 n_3 U_{n_2,n_3}^{(2,3)}}{n(n-1)} - \frac{n_2(2+n_3) U_{n_2}^{(2)}}{n(n-1)} - \\ &\quad \frac{n_3(2+n_2) U_{n_3}^{(3)}}{n(n-1)}. \end{aligned}$$

where $U_{n_k}^{(k)} = \binom{n_k}{2}^{-1} \sum_{1 \leq i < j \leq k} \phi(X_{ki}, X_{kj})$ and
$U_{n_g, n_{g'}} = \frac{1}{n_g n_{g'}} \sum_{i=1}^{n_g} \sum_{i=1}^{n_{g'}} \phi(X_{gi}, X_{g'j})$.
The Hoeffding decomposition of $B_n$ is

$$
\begin{aligned}
B_n &= \frac{2n_2}{n(n-1)} \left[ \frac{1}{n_2} \sum_{i=1}^{n_2} \phi(X_1, X_{2i}) \right] + \frac{2n_3}{n(n-1)} \left[ \frac{1}{n_3} \sum_{j=1}^{n_3} \phi(X_1, X_{3j}) \right] + \\
&\quad + \frac{2n_2 n_3}{n(n-1)} \left[ \frac{1}{n_2 n_3} \sum_{i=1}^{n_2} \sum_{j=1}^{n_3} \phi(X_{2i}, X_{3j}) \right] \\
&\quad - \frac{n_2(2+n_3)}{n(n-1)} \left[ \binom{n_2}{2}^{-1} \sum_{1 \leq i < j \leq n_2} \phi(X_{2i}, X_{2j}) \right] \\
&\quad - \frac{n_3(2+n_2)}{n(n+1)} \left[ \binom{n_3}{2}^{-1} \sum_{1 \leq i < j \leq n_3} \phi(X_{3i}, X_{3j}) \right] \\
&= \frac{2}{n(n-1)} \sum_{i=1}^{n_2} \phi(X_1, X_{2i}) + \frac{2}{n(n-1)} \sum_{j=1}^{n_3} \phi(X_1, X_{3j}) + \\
&\quad + \frac{2}{n(n-1)} \sum_{i=1}^{n_2} \sum_{j=1}^{n_3} \phi(X_{2i}, X_{3j}) \\
&\quad - \frac{n_2(2+n_3)}{n(n-1)} \frac{2}{n_2(n_2-1)} \sum_{1 \leq i < j \leq n_2} \phi(X_{2i}, X_{2j}) \\
&\quad - \frac{n_3(2+n_2)}{n(n-1)} \frac{2}{n_3(n_3-1)} \sum_{1 \leq i < j \leq n_3} \phi(X_{3i}, X_{3j}).
\end{aligned}
$$

Is known from the theory of U-statistics (see [Hoeffding, 1948]) that the kernel $\phi(,)$ can be expressed as sum of orthogonal components, $\phi(X_i, X_j) = \psi_1(X_i) + \psi_1(X_j) + \psi_2(X_i, X_j) + \theta$, where $\psi_1(X_i) = \mathbb{E}[\phi(X_i, X_j)|X_i]$, and $\psi_2(X_i, X_j) = \mathbb{E}[\phi(X_i, X_j)|X_i, X_j]$.

Then,

$$
\begin{aligned}
B_n &= \\
&\quad \frac{2}{n(n-1)} \sum_{i=1}^{n_2} [\psi_1(X_1) + \psi_1(X_{2i}) + \psi_2(X_1, X_{2i}) + \theta] + \\
&\quad + \frac{2}{n(n-1)} \sum_{j=1}^{n_3} [\psi_1(X_1) + \psi_1(X_{3j}) + \psi_2(X_1, X_{3j}) + \theta] + \\
&\quad + \frac{2}{n(n-1)} \sum_{i=1}^{n_2} \sum_{j=1}^{n_3} [\psi_1(X_{2i}) + \psi_1(X_{3j}) + \psi_2(X_{2i}, X_{3j}) + \theta] + \\
&\quad + \left(-\frac{2(2+n_3)}{n(n-1)(n_2-1)}\right) \sum_{1 \leq i < j \leq n_2} [\psi_1(X_{2i}) + \psi_1(X_{2j}) \\
&\quad + \psi_2(X2i, X_{2j}) + \theta] + \\
&\quad + \left(-\frac{2(2+n_2)}{n(n-1)(n_3-1)}\right) \sum_{1 \leq i < j \leq n_3} [\psi_1(X_{3i}) + \psi_1(X_{3j}) + \\
&\quad + \psi_2(X_{3i}, X_{3j}) + \theta] \\
&= \theta \left[\frac{2n_2}{n(n-1)} + \frac{2n_3}{n(n-1)} + \frac{2n_2 n_3}{n(n-1)} - \frac{2(2+n_3)}{n(n-1)(n_2-1)} \frac{n_2(n_2-1)}{2}\right. \\
&\quad \left. - \frac{2(2+n_2)}{n(n-1)(n_3-1)} \frac{n_3(n_3-1)}{2}\right] + \psi_1(X_1) \left[\frac{2n_2}{n(n-1)} + \right. \\
&\quad \left. + \frac{2n_3}{n(n-1)}\right] + \sum_{i=1}^{n_2} \psi_1(X_{2i}) \left[\frac{2}{n(n-1)} + \frac{2n_3}{n(n-1)} - \frac{2(2+n_3)}{n(n-1)}\right] + \\
&\quad + \sum_{j=1}^{n_3} \psi_1(X_{3j}) \left[\frac{2}{n(n-1)} + \frac{2n_2}{n(n-1)} - \frac{2(2+n_2)}{n(n-1)}\right] \\
&\quad + \frac{2}{n(n-1)} \sum_{i=1}^{n_2} \psi_2(X_1, X_{2i}) + \frac{2}{n(n-1)} \sum_{j=1}^{n_3} \psi_2(X_1, X_{3j}) + \\
&\quad + \frac{2}{n(n-1)} \sum_{i=1}^{n_2} \sum_{j=1}^{n_3} \psi_2(X_{2i}, X_{3j}) \\
&\quad - \frac{2(2+n_3)}{n(n-1)(n_2-1)} \sum_{1 \leq i < j \leq n_2} \psi_2(X_{2i}, X_{2j}) \\
&\quad - \frac{2(2+n_3)}{n(n-1)(n_3-1)} \sum_{1 \leq i < j \leq n_3} \psi_2(X_{3i}, X_{3j}) \\
&= \psi_1(X_1) \left(\frac{2}{n}\right) + \sum_{i=1}^{n_2} \psi_1(X_{2i}) \left(\frac{2 + 2n_3 - 4 - 2n_3}{n(n-1)}\right) + \\
&\quad + \sum_{j=1}^{n_3} \psi_1(X_{3j}) \left[\frac{2 + 2n_2 - 4 - 2n_2}{n(n-1)}\right] + \frac{2}{n(n-1)} \sum_{i=1}^{n_2} \psi_2(X_1, X_{2i}) + \\
&\quad + \frac{2}{n(n-1)} \sum_{j=1}^{n_3} \psi_2(X_1, X_{3j}) + \frac{2}{3 n(n-1)} \sum_{i=1}^{n_2} \sum_{j=1}^{n_3} \psi_2(X_{2i}, X_{3j}) \\
&\quad - \frac{2(2+n_3)}{n(n-1)(n_2-1)} \sum_{1 \leq i < j \leq n_2} \psi_2(X_{2i}, X_{2j}) \\
&\quad - \frac{2(2+n_2)}{n(n-1)(n_3-1)} \sum_{1 \leq i < j \leq n_3} \psi_2(X_{3i}, X_{3j})
\end{aligned}
$$

Thus, the Hoeffding decomposition of $B_n$ for size one group case is

$$B_n = \frac{2}{n}\left[\psi_1(X_1) - \frac{1}{n-1}\sum_{i=1}^{n_2}\psi_1(X_{2i}) - \frac{1}{n-1}\sum_{j=1}^{n_3}\psi_1(X_{3j}) + \right.$$
$$+ \frac{1}{n-1}\sum_{i=1}^{n_2}\psi_2(X_1, X_{2i}) + \frac{1}{n-1}\sum_{j=1}^{n_3}\psi_2(X_1, X_{3j}) +$$
$$+ \frac{1}{n-1}\sum_{i=1}^{n_2}\sum_{j=1}^{n_3}\psi_2(X_{2i}, X_{3j}) - \frac{(2+n_3)}{(n-1)(n_2-1)}\sum_{1\leq i<j\leq n_2}\psi_2(X_{2i}, X_{2j})$$
$$\left. - \frac{(2+n_2)}{(n-1)(n_3-1)}\sum_{1\leq i<j\leq n_3}\psi_2(X_{3i}, X_{3j})\right]$$

## S1.1 Finite sample properties of $B_n$

Let $\mathrm{E}\left[\phi(X_g, X_g)\right] = \theta_g$ and $\mathrm{E}\left[\phi(X_g, X_{g'})\right] = \theta_{gg'}$, then

$$\begin{aligned}
\mathrm{E}(B_n) &= \frac{2n_2\theta_{12}}{n(n-1)} + \frac{2n_3\theta_{13}}{n(n-1)} + \frac{2n_2n_3\theta_{23}}{n(n-1)} \\
&\quad - \frac{2(2+n_3)}{n(n-1)(n_2-1)}\frac{n_2(n_2-1)}{2}\theta_2 \\
&\quad - \frac{2(2+n_2)}{n(n-1)(n_3-1)}\frac{n_3(n_3-1)}{2}\theta_3 \\
&= \frac{1}{n(n-1)}\left[2n_2\theta_{12} + 2n_3\theta_{13} + 2n_2n_3\theta_{23} - n_2(2+n_3)\theta_2 \right. \\
&\quad \left. - n_3(2+n_2)\theta_3\right] \\
&= \frac{1}{n(n-1)}\left[n_2(2\theta_{12} - 2\theta_2) + n_3(2\theta_{13} - 2\theta_3) + \right. \\
&\quad \left. + n_2n_3(\theta_{23} - \theta_2) + n_2n_3(\theta_{23} - \theta_3)\right]
\end{aligned}$$

Under the null hypothesis $H_0$, $\theta_g = \theta_{gg'}$ and clearly $E(B_n) = 0$. Under the alternative $H_1$, $\mathbb{E}(B_n) > 0$ since we have $\theta_{gg'} > \theta_g$, for all $g \neq g' \in \{1,2,3\}$. This condition was already required in the work of [Valk and Cybis, 2020].

For accessing the $B_n$'s variance we handle with Hoeffding decomposition of $\frac{n}{2}B_n$ and otain $\mathrm{Var}\left(\frac{n}{2}B_n\right)$. It follows that

4

$$\frac{n}{2}B_n = \psi_1(X_1) - \frac{1}{n-1}\sum_{i=1}^{n_2}\psi_1(X_{2i}) - \frac{1}{n-1}\psi_1(X_{3j}) +$$
$$+ \frac{1}{n-1}\sum_{i=1}^{n_2}\psi_2(X_1, X_{2i}) + \frac{1}{n-1}\sum_{j=1}^{n_3}\psi_2(X_1, X_{3j}) +$$
$$+ \frac{1}{n-1}\sum_{i=1}^{n_2}\sum_{j=1}^{n_3}\psi_2(X_{2i}, X_{3j})$$
$$- \frac{(2+n_3)}{(n-1)(n_2-1)}\sum_{1\leq i<j\leq n_2}\psi_2(X_{2i}, X_{3j})$$
$$- \frac{(2+n_2)}{(n-1)(n_3-1)}\sum_{1\leq i<j\leq n_3}\psi_2(X_{3i}, X_{3j})$$

Define $\tau_1^2 = \text{Var}\left[\psi_1(X_1)\right]$ and $\tau_2^2 = \text{Var}\left[\psi_2(X_1, X_2)\right]$. Then, under $H_0$ when we have a size one group

$$\begin{aligned}
\text{Var}\left(\frac{n}{2}B_n\right) &= \tau_1^2 + \left(\frac{1}{n-1}\right)^2\sum_{j=1}^{n_2}\tau_1^2 + \left(\frac{1}{n-1}\right)^2\sum_{j=1}^{n_3}\tau_1^2 + \\
&+ \left(\frac{1}{n-1}\right)^2\sum_{i=1}^{n_2}\tau_2^2 + \left(\frac{1}{n-1}\right)^2\sum_{j=1}^{n_3}\tau_2^2 + \\
&+ \left(\frac{1}{n-1}\right)^2\sum_{i=1}^{n_2}\sum_{j=1}^{n_3}\tau_2^2 + \\
&+ \left[\frac{(2+n_3)}{(n-1)(n_2-1)}\right]^2\sum_{1\leq i<j\leq n_2}\tau_2^2 + \\
&+ \left[\frac{(2+n_2)}{(n-1(n_3-1)}\right]^2\sum_{1\leq i<j\leq n_3}\tau_2^2 \\
&= \tau_1^2 + \frac{1}{n-1}\tau_1^2 + \frac{1}{n-1}\tau_1^2 + \\
&+ \frac{n_2 n_3}{(n-1)^2}\tau_2^2 + \frac{n_2(2+n_3)^2}{2(n_2-1)(n-1)^2}\tau_2^2 + \frac{n_3(2+n_2)^2}{2(n_3-1)(n-1)^2}\tau_2^2 \\
&= \tau_1^2\frac{n}{n-1} + \tau_2^2\left[\frac{1}{n-1} + \frac{n_2 n_3}{(n-1)^2} + \right. \\
&\left. + \frac{n_2(2+n_3)^2}{2(n_2-1)(n-1)^2} + \frac{n_3(2+n_2)^2}{2(n_3-1)(n-1)^2}\right]
\end{aligned}$$

Therefore

$$\text{Var}(B_n) = \tau_1^2\left[\frac{4}{n(n-1)}\right] + \tau_2^2\left[\frac{4}{n^2(n-1)} + \frac{4n_2 n_3}{n^2(n-1)^2} + \right.$$
$$\left. + \frac{2n_2(2+n_3)^2}{n^2(n_2-1)(n-1)^2} + \frac{2n_3(2+n_2)^2}{n^2(n_3-1)(n-1)^2}\right] \quad \text{(S.2)}$$

Note that $n_3 = n - 1 - n_2$, then we can rewrite $\text{Var}(B_n)$ as

$$\text{Var}(B_n) = \eta_1(n)\tau_1^2 + \eta_2(n;n_2)\tau_2^2 \quad \text{(S.3)}$$

## S1.2 Asymptotic properties $B_n$'s variance

We show that

$$\text{Var}\left(\frac{n}{2}B_n\right) = \tau_1^2 \frac{n}{n-1} + \tau_2^2\left[\frac{1}{n-1} + \frac{n_2 n_3}{(n-1)^2} + \frac{n_2(2+n_3)^2}{2(n_2-1)(n-1)^2} + \right.$$
$$\left. + \frac{n_3(2+n_2)^2}{2(n_3-1)(n-1)^2}\right].$$

Note that

$$\text{Var}\left(\frac{n}{2}B_n\right) = \tau_1^2 O(1) + \tau_2^2\left[O(n^{-1}) + O(1) + O(1) + O(1)\right]$$
$$= O(1).$$

Let $\tau_n = \frac{n}{2}\sqrt{\text{Var}(B_n)}$.

A simple consequence is that $\text{Var}(\frac{n}{2}B_n) = \frac{n^2}{4}\text{Var}(B_n) = O(1)$. Thus, it follows that
$$\tau_n = \frac{n}{2}\sqrt{\text{Var}(B_n)} = O(1).$$

# S2 The clustering method *uclust*3

The algorithm for the clustering method *uclust*3, introduced in the Section 3.1 of the main manuscript, can be described as follows. We apply the homogeneity test on the dataset and if it returns "non homogeneous", we then find the partition $\{G_1^\star, G_2^\star, G_3^\star\}$ that maximizes $B_n$ and set $n_1^\star$ as the smallest subgroup size. Among all possible configurations in which one of the groups has size one, we find the configuration $\{G_1^{\star 1}, G_2^{\star 1}, G_3^{\star 1}\}$ that maximizes $B_n$ and set this $B_n$ value as $B_n^1$. If $\{G_1^\star, G_2^\star, G_3^\star\}$ is a significant configuration and $B_n > B_n^1$, we have found our optimal partition. If $\{G_1^\star, G_2^\star, G_3^\star\}$ is a significant partition and $B_n < B_n^1$ with $B_n^1$ significant, then $\{G_1^{\star 1}, G_2^{\star 1}, G_3^{\star 1}\}$ is our optimal partition.

However, if this maximal $B_n$ comes from a non significant partition $\{G_1^\star, G_2^\star, G_3^\star\}$, then there are no other significant partitions in configurations

with smaller group size between 2 and $n_1^\star$. The restricted search is done on subgroups with sizes larger than $n_1\star$, until it finds the significant partition and compares with $B_n^1$, returning the configuration with maximum significant $B_n$. By exploring this insight, we built the following clustering algorithm based on restricted optimization problems.

---

*uclust*3 **Algorithm**: Finds the data partition that maximizes $B_n$ in the universe of all significant partitions

---
**Input**: Data **X**
**Output**: Partition $\{G_1^\star, G_2^\star, G_3^\star\}$
01: Apply homogeneity test to **X**
02: **if** Accept $H_0$
03:   **Return** $G_1^\star = \emptyset$, $G_2^\star = \emptyset$ and $G_3^\star = \{\mathbf{X}_1, \ldots, \mathbf{X}_n\}$
04: **else**
05:   find $G_1^\star$, $G_2^\star$ and $G_3^\star$ that optimize $B_n$. Set this results as $B_n$
06:   For $G_1^{\star 1}$ of size one, find $G_1^{\star 1}$, $G_2^{\star 1}$ and $G_3^{\star 1}$ that optimize $B_n$.
      Set this results as $B_n^1$
07:   **If** $B_n$ is significant
08:     **If** $B_n < B_n^1$ and $B_n^1$ is significant, $G_1^\star = G_1^{\star 1}$, $G_2^\star = G_2^{\star 1}$ and $G_3^\star = G_3^{\star 1}$
9:   **else**
10:     Set $G_1^\star$ size $(n_1^\star)$ as the smallest size among $G_1^\star$, $G_2^\star$ and $G_3^\star$
11:     **while** $\{G_1^\star, G_2^\star, G_3^\star\}$ is not significant partitions
12:       **while** $\{G_1^\star, G_2^\star, G_3^\star\}$ is not significant partitions.
13:         $n_2 \in \{(n_1^\star + 1), \ldots, (n - 2n_1^\star + 1)\}$, find
            $G_1$, $G_2$ and $G_3$ that optimize $B_n$ for subgroup size and set
            $G_1^\star = G_1$, $G_2^\star = G_2$ and $G_3^\star = G_3$
14:       $n_1^\star = n_1^\star + 1$
15:     **Compare** $B_n$ and $B_n^1$ and do **08**
16:   **Return** $\{G_1^\star, G_2^\star, G_3^\star\}$

---

The multiple optimization subproblems in the *uclust*3 algorithm are solved through a cyclic coordinate ascent algorithm repeated multiple times with random starting clusters to account for local optima.

## S3  Simulations Studies

In this section we present simulation studies in order to evaluate some aspects of our proposed methodology, a complementary material for the simulation studies shown at Section 4 of the main manuscript. At first we evaluate the size and power of the proposed *utest* for homogeneity of three groups.

### S3.1  Simulations for the *utest*

We present here a simulation study to evaluate the performance of the *utest* for three groups. The data was simulate as shown at Section 4.1 of the main manuscript, but this time without group of size one. The groups $G_1$ and $G_2$ were set with the same size $n_1 = n_2 = \lfloor n/3 \rfloor$.

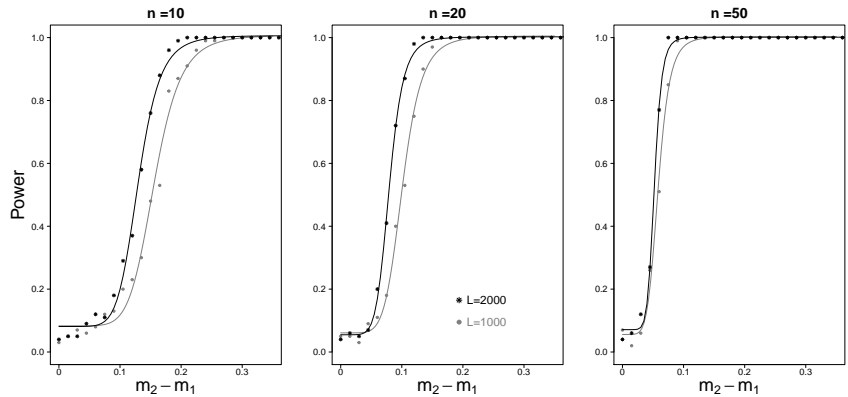

Figure S1: Power curves of *utest* for two dimension $L = 1000$ (gray) and $L = 2000$ (black) for 100 replications of each scenario of $n$ with $n_1 = n_2 = \lfloor \frac{n}{3} \rfloor$ and $n_3 = n - n_1 - n_2$.

### S3.2 Simulations for homogeneity test *uclust3*

Similarly to Section 4.2 of the main manuscript we used simulation studies to evaluate the homogeneity test.

#### S3.2.1 Size of homogeneity test *uclust3*

First the data were simulated following the same distribution. All elements from the $n \in \{10, 20, 30, 40, 50, 100\}$ vectors with dimension $L \in \{1000, 2000\}$ were generated following a Normal distribution with mean 0 and variance 1. The homogeneity test was applied to the dataset and observed if the null hypothesis was rejected or not. This process was replicated 100 times and the size of the test can be seen at the following table.

Table S1: Size of homogeneity test *uclust3*

| $n$ | Dimension $L$ | |
| --- | --- | --- |
| | 1000 | 2000 |
| 10 | 0.01 | 0.01 |
| 20 | 0 | 0 |
| 30 | 0.01 | 0 |
| 40 | 0.02 | 0 |
| 50 | 0.03 | 0.03 |
| 100 | 0.14 | 0.03 |

8

### S3.2.2 Power of homogeneity test *uclust*3

In order to evaluate the power of our proposed homogeneity test *uclust*3 we simulate data from independent normally distributed vectors divided in three groups $G_1$, $G_2$ and $G_3$. The $L$ dimensional vectors in $G_1$ are generated from a independent and identically normal with mean $m_1 = 0$ and variance 1. The elements of the vectors in $G_2$ and $G_3$ have the same properties with mean $m_2$ and $m_3$, respectively. For each sample size $n$ in $\{10, 20, 50\}$, the $G_1$ and $G_2$ group sizes $n_1$ and $n_2$ were chosen so that we had a central configuration, in which the groups have approximately the same number of elements and a extremely configuration in which one of the groups has only two elements and the other has $n/2$ elements. Naturally the third group size's is defined as $n_3 = n - n_1 - n_2$.

Table S2: Power of homogeneity test *uclust*3

| $n$ | $(m_2, m_3)$ | $(n_1, n_2)$ | Dimension $L$ 1000 | 2000 |
|---|---|---|---|---|
| 10 | (0.25, 0.5) | (2, 5) | 0.21 | 0.31 |
| | | (3, 3) | 0.06 | 0.09 |
| | (0.5, 1) | (2, 5) | 0.21 | 0.24 |
| | | (3, 3) | 0.02 | 0.02 |
| 20 | (0.25, 0.5) | (2, 10) | 1 | 1 |
| | | (6, 6) | 1 | 1 |
| | (0.5, 1) | (2, 10) | 1 | 1 |
| | | (6, 6) | 1 | 1 |
| 50 | (0.25, 0.5) | (2, 25) | 1 | 1 |
| | | (16, 16) | 1 | 1 |
| | (0.5, 1) | (2, 25) | 1 | 1 |
| | | (16, 16) | 1 | 1 |

### S3.3 Simulations for finding correct clusters comparing with the *kmeans*

We complement the simulations study in Section 4.3 by performing a comparison between *uclust*3 method and *kmeans* clustering algorithm for the case where we have a size one group.

Table S3: Comparison of mean ARI and standard deviation (Sd) of the accuracy in clustering of *kmeans* and *uclust*3 methods with a size one group.

| $n$ | $(m_2, m_3)$ | $(n_2)$ | Method | Dimension $L$ | | | |
|---|---|---|---|---|---|---|---|
| | | | | 1000 | | 2000 | |
| | | | | Mean | Sd | Mean | Sd |
| 10 | (0.25, 0.5) | 2 | *kmeans* | 0.44 | 0.03 | 0.48 | 0.05 |
| | | | *uclust*3 | 0.47 | 0.03 | 0.5 | 0.06 |
| | | 5 | *kmeans* | 0.66 | 0.02 | 0.73 | 0.03 |
| | | | *uclust*3 | 0.74 | 0.03 | 0.79 | 0.03 |
| | (0.5, 1) | 2 | *kmeans* | 0.86 | 0.07 | 0.94 | 0.04 |
| | | | *uclust*3 | 0.75 | 0.1 | 0.82 | 0.08 |
| | | 5 | *kmeans* | 0.93 | 0.03 | 0.97 | 0.01 |
| | | | *uclust*3 | 0.99 | 0 | 1 | 0 |
| 20 | (0.25, 0.5) | 2 | *kmeans* | 0.34 | 0.02 | 0.33 | 0.02 |
| | | | *uclust*3 | 0.33 | 0.01 | 0.36 | 0.02 |
| | | 10 | *kmeans* | 0.73 | 0 | 0.77 | 0.01 |
| | | | *uclust*3 | 0.73 | 0.01 | 0.74 | 0.01 |
| | (0.5, 1) | 2 | *kmeans* | 0.62 | 0.12 | 0.83 | 0.09 |
| | | | *uclust*3 | 0.67 | 0.12 | 0.98 | 0.01 |
| | | 10 | *kmeans* | 0.95 | 0.01 | 0.97 | 0.01 |
| | | | *uclust*3 | 0.92 | 0.02 | 1 | 0 |
| 50 | (0.25, 0.5) | 2 | *kmeans* | 0.17 | 0.01 | 0.17 | 0.01 |
| | | | *uclust*3 | 0.15 | 0 | 0.15 | 0 |
| | | 25 | *kmeans* | 0.75 | 0 | 0.76 | 0 |
| | | | *uclust*3 | 0.74 | 0 | 0.74 | 0 |
| | (0.5, 1) | 2 | *kmeans* | 0.22 | 0.02 | 0.35 | 0.1 |
| | | | *uclust*3 | 0.18 | 0.02 | 0.41 | 0.15 |
| | | 25 | *kmeans* | 0.9 | 0.02 | 0.94 | 0.01 |
| | | | *uclust*3 | 0.82 | 0.01 | 0.99 | 0 |

Over the 100 replications observing the different scenarios we can conclude that both methods compete, alternating in the presentation of the best results.

## S4 Finding correct clusters comparing *uclust*3 and *uhclust* in a presence of an outlier

We complement the simulation study presented in Section 4.4 of the manuscript considering here only groups larger than 2. Figures S2 and S3 report curves of proportion times that the algorithms found significant separation and correct groups considering different values of $m_2 - m_1$ varying on the $x$ axis, sample size $n$ taking values in $\{10, 20, 50\}$ and dimension $L = 1000$ and $L = 2000$. The

results are based on 50 repetitions.

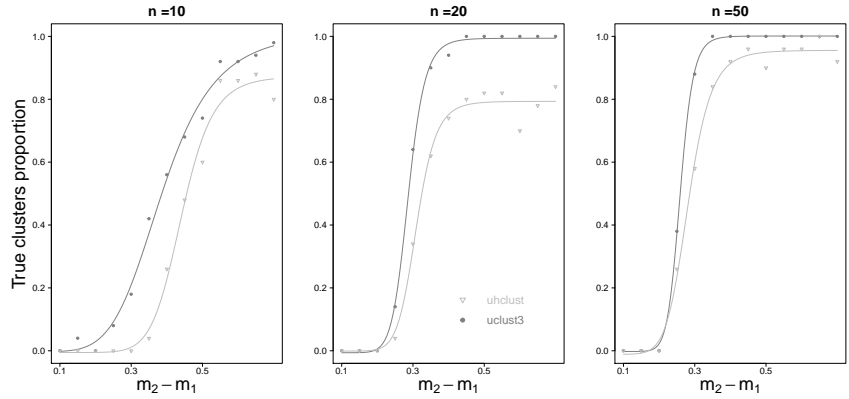

Figure S2: True cluster proportion curves of *uclust*3 and *uhclust* for dimension $L = 1000$ with 50 replications of each scenario of $n$ with $\alpha = 0.05$.

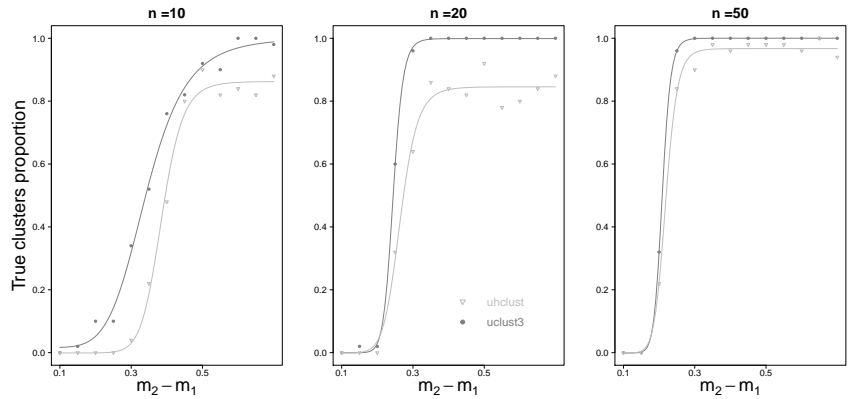

Figure S3: True cluster proportion curves of *uclust*3 and *uhclust* for dimension $L = 2000$ with 50 replications of each scenario of $n$ with $\alpha = 0.05$.

## S5 Application

In the interest of evaluating the performance of the proposed method *uclust*3 comparing with *uhclust* and *sigclust* we consider an image group configuration with an outlier. The data are the same as described in Section 5 in the main manuscript. We randomly select 1 image from Tony Blair and 10 images from each other public figure in the above cited dataset and run *uhclust*, *sigclust*

11


\end{document}